\documentclass[12pt]{article}

\def \be {\begin{equation}}
\def \ee {\end{equation}}
\def \bea {\begin{eqnarray}}
\def \eea {\end{eqnarray}}
\newcommand{\upd}{\mathrm{d}}

\usepackage[utf8]{inputenc}
\usepackage{amsmath}
\usepackage{amssymb}
\usepackage{amsthm}
\usepackage{slashed}
\usepackage{graphicx}
\usepackage[backend=biber, style=numeric, sorting=none]{biblatex}
\renewbibmacro{in:}{%
  \ifentrytype{article}{}{\printtext{\bibstring{in}\intitlepunct}}}

\usepackage{hyperref}

\newtheorem{theorem}{Theorem}[section]
\newtheorem{lemma}[theorem]{Lemma}
\newtheorem{proposition}[theorem]{Proposition}

\theoremstyle{definition}
\newtheorem{definition}[theorem]{Definition}

\theoremstyle{remark}

\newcommand{\dVol}{\mathrm{d}\textit{vol}}
\DeclareMathOperator{\spn}{span}

\title{Wave propagation on microstate geometries}

\begin{document}
\author{Joe Keir \\ \\
{\small Mathematical Institute, University of Oxford, Andrew Wiles Building,} \\
{\small Radcliffe Observatory Quarter (550), Woodstock Road, Oxford, OX2 6GG} \vspace{5mm} \\ 
\small{joseph.keir@maths.ox.ac.uk}}

\maketitle

\begin{abstract}
 Supersymmetric microstate geometries were recently conjectured \cite{Eperon2016} to be nonlinearly unstable due to numerical and heuristic evidence, based on the existence of very slowly decaying solutions to the linear wave equation on these backgrounds. In this paper, we give a thorough mathematical treatment of the linear wave equation on both two and three charge supersymmetric microstate geometries, finding a number of surprising results. In both cases we prove that solutions to the wave equation have uniformly bounded local energy, despite the fact that three charge microstates possess an ergoregion; these geometries therefore avoid Friedman's ``ergosphere instability'' \cite{Friedman1978}. In fact, in the three charge case we are able to construct solutions to the wave equation with local energy that neither grows nor decays, although this data must have nontrivial dependence on the Kaluza-Klein coordinate. In the two charge case we construct quasimodes and use these to bound the uniform decay rate, showing that the only possible uniform decay statements on these backgrounds have very slow decay rates. We find that these decay rates are sub-logarithmic, verifying the numerical results of \cite{Eperon2016}. The same construction can be made in the three charge case, and in both cases the data for the quasimodes can be chosen to have trivial dependence on the Kaluza-Klein coordinates.
\end{abstract}

\section{Introduction}

\subsection{Microstate geometries}

``Microstate geometries'' are a large family of solutions to type IIB supergravity with several interesting features (\cite{Maldacena2002} \cite{Balasubramanian2001}, \cite{Lunin:2002iz} \cite{Giusto2004} \cite{Giusto2004-2} \cite{Bena2006} \cite{Berglund2006}). They are smooth and ``asymptotically Kaluza-Klein'': near infinity, they approach the product of five dimensional Minkowski space with five compact dimensions\footnote{Note that four of these compact dimensions will play no role whatsoever in this paper: they could be included, but would not affect our results. Alternatively, we can simply consider the corresponding six dimensional spacetime, which will be the approach taken in this paper.}. They do not possess a black hole region or a horizon, although they do share several geometric features with black hole solutions, including the ``trapping'' of null geodesics and the possibility of possessing an ergoregion. In addition, they can exhibit an ``evanescent ergosurface''\cite{Gibbons2013}: a timelike submanifold on which an otherwise timelike Killing vector field becomes null. The ``fuzzball proposal'' \cite{Mathur:2005zp} conjectures that they provide a geometric description of certain quantum microstates of black holes, providing further motivation for their study.

A natural question to ask regarding microstate geometries is whether they are classically stable, i.e.\ does there exist a suitable function space for the initial data, and an open set around the initial data for the microstate solution in question, such that the future evolution of all initial data in this set remains suitably ``close'' to the microstate geometry? In general, this is a very difficult question to address in a nonlinear field theory such as supergravity, but we may begin to address it by studying suitable \emph{linear} equations on a fixed geometric background. In this paper we shall study the behaviour of linear waves, that is, solutions of the equation
\begin{equation}
\label{equation wave equation}
 \Box_g u = 0
\end{equation}
on microstate geometries, where $g$ is the metric of the corresponding microstate geometry. Note that there are certain linearised supergravity fields which obey this equation \cite{Cardoso:2007ws}, but it can also be considered as a toy model for the linearisation of the equation of motion for the metric, neglecting both the tensorial structure and the coupling to matter.

A large family of microstate solutions have been found, and in particular, the geometries may be either ``supersymmetric'' or not; the supersymmetric microstate geometries possess a global, null Killing vector field \cite{Gutowski:2003rg}, while the non-supersymmetric microstates do not \cite{Jejjala2005}. In \cite{Cardoso:2005gj}, the linear stability of non-supersymmetric microstate geometries was studied. Both heuristic and numerical evidence was presented, all of which points to a linear instability of these geometries. This instability can be understood as an instance of the ``Friedman instability'' or ``ergosphere instability'' \cite{Friedman1978}: the non-supersymmetric microstate geometries possess an ergoregion but no horizon, meaning that perturbations can be localised within the ergoregion and cannot decay in time. In these circumstances Friedman provided a heuristic argument indicating that the local energy of solutions should \emph{grow} in time; this was investigated numerically in, for example (\cite{Cardoso:2007az} \cite{Comins211} \cite{Yoshida}) and was very recently proved rigorously in \cite{Moschidis:2016zjy}, under certain conditions. Although these conditions do not apply to microstate geometries\footnote{Specifically, \cite{Moschidis:2016zjy} requires that the spacetime is asymptotically flat, rather than asymptotically Kaluza-Klein. Note, however, that in the case of non-supersymmetric microstate geometries, we can restrict to waves with trivial dependence on the Kaluza-Klein direction, and the argument of \cite{Moschidis:2016zjy} does apply to these solutions.}, \cite{Cardoso:2005gj} in fact produced evidence for exponentially growing solutions to the linearised equations of motion in the non-supersymmetric microstate geometries. Note that there are also examples of instabilities associated with ergoregions in spacetimes \emph{with} horizons (see, for example,\cite{Shlapentokh-Rothman2014} and \cite{Moschidis:2016wew}, the heuristic work of \cite{Damour:1976kh} \cite{Zouros} \cite{Detweiler:1980uk} and the numerical results in \cite{Dolan:2007mj} \cite{Dolan:2012yt}).

On the other hand, the presence of an additional causal Killing vector field in the supersymmetric microstate geometries might suggest that they have better stability properties than their non-supersymmetric counterparts (see, for example, the comments in \cite{Cardoso:2005gj}). However, in the very recent work \cite{Eperon2016}, both heuristic and numerical evidence was provided which indicates that these geometries might \emph{also} be unstable, but in contrast to their non-supersymmetric counterparts, this instability is conjectured to be nonlinear in nature. The source of this instability was identified as the presence of \emph{stably trapped} null geodesics near the evanescent ergosurface, i.e.\ null geodesics that remain trapped in a bounded area of space, in such a way that nearby null geodesics are also trapped. Indeed, this stable trapping was shown to be a generic feature of spacetimes possessing evanescent ergosurfaces. In addition, heuristic arguments for instability were given that made use of the unusual fact that there are stably trapped null geodesics \emph{with zero energy} measured with respect to an asymptotically timelike Killing vector field; these are the null geodesics that rule the evanescent ergosurface. Note that both \cite{Eperon2016} and the current paper focus only on a special class of supersymmetric microstate geometries, rather than the more general class constructed in \cite{Lunin:2002iz}, which posses fewer symmetries than the spacetimes we consider.

\subsection{Stable trapping and slow decay}

Previous studies of wave propagation on spacetimes with stably trapped null geodesics have shown that linear waves on these backgrounds decay very slowly (\cite{Holzegel2014, Keir2016, Benomio2018}). This suggests that nonlinear instabilities might be present, since waves might have time to ``clump'' sufficiently for nonlinear effects to play a role before dispersion can occur. In all of the references given above, the decay was found to be no faster than ``logarithmic'', that is, there is some open region $U$ and some norm of the initial data $E^{(1)}_0(u)$, depending only on the field $u$ and its first derivative, such that, for solutions $u$ to the wave equation \eqref{equation wave equation} with, say, Schwartz initial data, there is some universal positive constant $C$ such that
\begin{equation}
\label{equation uniform decay slower than log}
 \limsup_{t \rightarrow \infty} \ \sup_{u} \ \log(2 + t) \frac{ || u ||_{H^1[U]} }{ \sqrt{ E^{(1)}_0(u) } } \geq C
\end{equation}
This shows that no uniform decay statement with a uniform rate of decay that is faster than logarithmic can hold. If we instead take norms of the initial data involving higher derivatives, then the factor of $\log(2 + t)$ needs to be replaced by a factor of $\left(\log(2+t)\right)^n$ for some power $n$. Note, however, that this decay rate is always slower than polynomial, for any finite $n$.

Interestingly, in \cite{Eperon2016} numerical evidence was found suggesting that, in microstate geometries, linear waves decay even slower. In particular, the function $\log(2+t)$ in \eqref{equation uniform decay slower than log} should be replaced by another function which grows even slower at large $t$: approximately at the rate $(\log t)/(\log \log t)$. This means that, in order to recover a comparable uniform decay rate, additional derivatives of the initial data must be included in the ``initial energy'' $E_0$. Note, however, that \cite{Eperon2016} used \emph{quasinormal modes} to demonstrate this fact, and these do not arise from compactly supported or Schwartz initial data, so the two results are not directly comparable (see, however, \cite{Gannot2015}). Note that there is an extremely extensive body of work regarding quasinormal modes in the physics literature (see e.g.\ \cite{Kokkotas1999}), and a growing mathematical literature on the subject (see e.g.\ \cite{Zworski2017} for a review), including some work on backgrounds with stably trapped null geodesics \cite{Gannot2012, Gannot2016}.

As mentioned above, the results of \cite{Holzegel2014} and \cite{Keir2016} established that no uniform decay statement with rate faster than logarithmic can hold on the spacetimes investigated, namely, Kerr-AdS and ultracompact neutron stars. These results were complemented by proofs (in \cite{Holzegel2013} and \cite{Keir2016} respectively) of the uniform decay of waves on those backgrounds. In other words, not only \emph{can} waves decay at a (uniform) rate no faster than logarithmic, but in fact, all waves with suitable initial data actually \emph{do} decay at least logarithmically. This should be compared with the classical result of Burq \cite{Burq1998}, establishing that the local energy of waves decays logarithmically in the exterior of any ``obstacles'' in Minkowski space (without restriction on the shape of the obstacles or the trapping of geodesics caused by the obstacles), as well as the recent theorem of Moschidis \cite{Moschidis2015}, showing that the same result holds on a very general class of spacetimes. Indeed, in both of these cases an estimate of the form
\begin{equation}
 || u ||_{H^1(U)} + || \partial_t u ||_{L^2(U)}  \leq C \left( \log(2 + t) \right)^{-m} \sqrt{ E^{(m)}_0(u) }
\end{equation}
holds, where $E^{(m)}_0(u)$ denotes the initial $m$-th order energy of the field $u$, which is a quantity involving up to $(m+1)$ derivatives of the initial data (suitably weighted). In this context, the indication in \cite{Eperon2016} that a slower-than-logarithmic rate of decay might hold in microstate geometries is extremely interesting, although we note again that the quasinormal modes used in \cite{Eperon2016} are not expected to lie in the suitably weighted energy space. Nevertheless, this result may be taken to be even more strongly indicative of a possible nonlinear instability than in the previously studied cases.

\subsection{Boundedness results}

In this paper we provide a thorough mathematical analysis of the behaviour of solutions to the linear wave equation on microstate geometries. As in \cite{Eperon2016}, we restrict attention to supersymmetric microstate geometries, and we also focus on the simplest examples of supersymmetric microstates (rather than the larger class of solutions constructed in, for example, \cite{Lunin:2002iz}). We examine both two \cite{Maldacena2002} and three charge \cite{Giusto2004} microstate geometries; geometrically, these are distinguished by the fact that the three charge geometries exhibit an ergoregion, whereas the two state geometries only exhibit an evanescent ergosurface.

One of the most basic questions we can ask about solutions to the wave equation is whether they are uniformly bounded, and due to the lack of a globally timelike Killing vector field in the microstate geometries this not straightforward. Indeed, the presence of an ergoregion in the three charge geometries, together with the heuristic arguments of Friedman \cite{Friedman1978} and the rigorous proof of Moschidis \cite{Moschidis:2016zjy} (the conditions of which, however, do not apply to microstate geometries) strongly suggests that solutions to the wave equation on three charge microstate geometries might not be uniformly bounded, and in fact, there might exist growing solutions. If this were true, then the very slowly decaying solutions observed in \cite{Eperon2016} would not be the ``worst'' solutions, and we would instead find the more familiar situation of solutions to the wave equation which grow in time, perhaps in the form of exponentially growing mode solutions. Note that the presence of an ergoregion was noticed already in \cite{Jejjala2005}, who also commented on the absence of a ``superradiant instability'' due to the lack of a horizon. 

Despite the considerations above, in section \ref{section boundedness} we prove that, in both the two and three charge microstate geometries, solutions to the wave equation with suitable initial data remain bounded for all time. Note that, in the three charge case, waves remain bounded \emph{despite the presence of an ergoregion}, and so the microstate geometries avoid Friedman's ``ergosphere instability''. In the three charge case our proof of boundedness relies crucially on the presence of the null Killing vector field in the supersymmetric geometries, and so does not apply to the non-supersymmetric geometries, which were previously found to suffer from an ergosphere instability \cite{Cardoso:2005gj}. Additionally, the proof we present ``loses derivatives'', i.e.\ we are only able to bound the local energy at future times by a ``higher order'' energy (involving more derivatives) initially. This means that the boundedness estimate is very unlikely to be of much use in a nonlinear setting, although it works well in the case of linear waves studied in this paper. We also note that our proof makes use of the additional symmetries of the geometries we consider, so it does not apply to all of the more general microstate geometries constructed in \cite{Cardoso:2005gj}.

In the two charge case, we also find that we can prove boundedness with a loss of derivatives. In this case, the asymptotically timelike Killing vector field is globally causal, but becomes null at the evanescent ergosurface, meaning that the corresponding energy degenerates there. We can contrast this with the case of the exterior of black holes: even in the relatively simple case of a Schwarzschild black hole, the asymptotically timelike Killing vector field becomes null on the event horizon, and so the corresponding energy degenerates there. One way\footnote{Another way to approach this was given in \cite{Kay} but this made use of a particular symmetry in the Schwarzschild spacetime.} to overcome this is to make use of the celebrated \emph{red shift effect}: we can modify the vector field so that it is no longer Killing, but we find that the error terms this introduces can themselves be bounded by the non-degenerate energy \cite{Dafermos2009}. However, in the microstate geometries the evanescent ergosurface is timelike: there is no local red shift effect (which is in some ways reminiscent of the case of an extremal black hole -- see \cite{Aretakis:2012ei}), and the presence of trapped null geodesics prevents us from obtaining a suitable ``integrated local energy decay estimate'' (\cite{Ralston}, \cite{Sbierski2013a}) which could be used to bound error terms.

Note that, in another work \cite{Keir:2018hnv}, we have shown that, for a broad class of spacetimes that includes the microstate geometries studied here, this loss of derivatives in the boundedness statement \emph{cannot be avoided}. In other words, it is not possible to bound the energy at some future time in terms of some kind of initial energy. Hence, the ``boundedness with a loss of derivatives'' result which we show here cannot be improved to a standard boundedness result.

\subsection{Non-decay and slow decay results}

Next, in section \ref{section non decay} we show that, on the three charge microstate geometries (which have an ergoregion) we can construct initial data with negative energy with respect to an asymptotically timelike Killing vector field. This follows from the work of Friedman \cite{Friedman1978}, but for completeness and clarity we give a more explicit construction on the three charge microstate geometries. Consequently, there exist solutions to the wave equation whose local energy does not decay in time. Combined with the boundedness result above, we conclude that the behaviour of generic solutions to the wave equation with suitable initial data is neither to decay nor grow over time. In particular, the local energy within the ergoregion will not decay over time, and yet there are no solutions with growing local energy.

The situation is different in the case of two charge microstate geometries, since these do not possess an ergoregion, but only an evanescent ergosurface. Thus we cannot use the construction of Friedman to find solutions that do not decay in time, but we can still prove boundedness in the same way as for the three charge case. Instead of showing that there are solutions which do not decay in time, in section \ref{section slow decay} we adapt the \emph{quasimode} construction, first used (in the context of general relativity) in \cite{Holzegel2014} (see also \cite{Keir2016, Benomio2018}), to construct very slowly decaying solutions. In fact, we are able to construct solutions\footnote{To be precise, we do not actually construct a solution which decays at this rate, but we do construct a sequence of solutions which decay at a rate arbitrarily close to this decay rate, for an arbitrarily long time.} which decay even more slowly than the solutions constructed in \cite{Holzegel2014} and \cite{Keir2016}, i.e.\ at a sub-logarithmic rate, verifying the numerical results of \cite{Eperon2016}. Note that these waves may be chosen to have trivial dependence on the compact directions, so the reason that the general logarithmic decay result of \cite{Moschidis2015} does not hold on supersymmetric microstate geometries is due to the fact that there does not exist a global, timelike Killing vector field on these geometries.

In addition, since we use quasimodes rather than quasinormal modes, we also improve the class of initial data leading to this slow decay rate, since the quasimodes we construct induce Schwartz initial data. In contrast, quasinormal modes do not even have finite energy on hypersurfaces which extend to spacelike infinity, although they do have finite energy on hypersurfaces extending to future null infinity, and the quasinormal modes constructed in \cite{Eperon2016} were found to be localised near the evanescent ergosurface. Nevertheless, this is an important point, since even in Minkowski space, we can construct solutions to the wave equation with arbitrarily slow decay, if we restrict only to initial data with finite energy\footnote{Specifically, we could use geometric optics, or the Gaussian beams of \cite{Sbierski2013a}, to construct solutions to the wave equation that are localised around an ingoing null ray, and then take a sequence of initial data such that this incoming ray is initially positioned at further and further distances from the origin.}. Hence, our quasimode construction not only verifies the slow decay rate found in \cite{Eperon2016}, but also confirms the expectation that this decay is caused by the local geometry of the microstate, and is not an artefact of the slow decay of the initial data towards infinity.

Our quasimode construction is the most technical part of this paper. Before discussing it further, we shall first give a brief overview of the role of quasimodes in the slow decay results of \cite{Holzegel2014} and \cite{Keir2016}. The idea is to first separate the wave equation, and then to note that the radial part of the wave equation involves an effective potential with a local minimum near some fixed radius. The effective potential also involves a factor of $\ell^2$, where $\ell$ is the angular frequency of the wave. The idea is then to construct approximate solutions near this local minimum by first constructing solutions to the corresponding Dirichlet problem, with boundary conditions imposed on either side of the local minimum. To reach these boundaries, the wave has to tunnel through the effective potential, and so we find that, near the boundaries, the wave is exponentially suppressed. Since the height of the potential scales with $\ell^2$, we find that the size of the wave near the boundaries behaves as $e^{-\ell}$. Hence, by smoothly cutting off the solution near these boundaries, we obtain approximate solutions to the wave equation, with errors that are exponentially small in $\ell$. This exponentially small error then leads directly to the logarithmic bound on the decay rate.

In showing that the decay rate on the two charge microstate geometries is even slower than logarithmic, the key observation (made in \cite{Eperon2016}) is the fact that, in these geometries, there are stably trapped null geodesics \emph{with zero energy} measured with respect to the asymptotically timelike Killing vector field. Together with the fact that the wave equation separates on the supersymmetric microstate geometries, this leads to a situation in which the effective potential for the radial part of the wave equation has a local minimum, at which the effective potential \emph{vanishes} to leading order at large $\ell$. We can use this fact by constructing quasimodes that are localised near this local minimum, but then exploiting the fact that both the height \emph{and the width} of the potential barrier which these waves must tunnel through scales with $\ell$; see figure \ref{figure}. Since the microstate geometries are asymptotically flat, the effective potential behaves asymptotically as $\ell^2 r^{-2}$, so the width of the potential barrier scales\footnote{We note in passing that this idea could be used to construct metrics on which the uniform decay rate for linear waves is arbitrarily slow, by ensuring that the effective potential decays slower in $r$, although these spacetimes would not be asymptotically flat.} as $\ell$. This then leads to quasimodes with errors that are super-exponentially suppressed in $\ell$, and that in turn leads to slower than logarithmic decay rates.

\begin{figure}[ht]
\centering
\includegraphics[width = 0.8\linewidth, keepaspectratio]{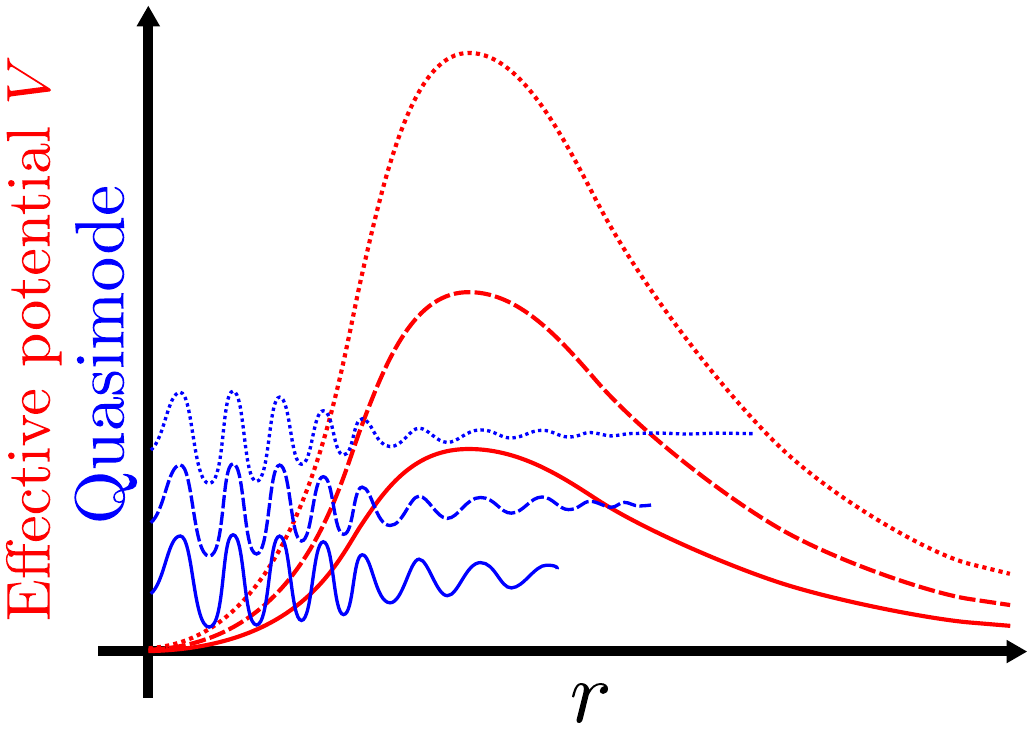}
\caption{A sketch of the regular part of the effective potential (red) at three increasing values of $\ell$ and the corresponding quasimodes (blue). The solid lines are associated with the lowest value of $\ell$, the dashed lines an intermediated value of $\ell$, and the dotted lines the highest value of $\ell$. Note that, at higher values of $\ell$, the potential barrier increases in height, and the quasimodes can be cut-off at larger values of $r$. In this kind of sketch the quasimodes will often be shifted up by their corresponding eigenvalue; in this case, the quasimodes all have very similar eigenvalues, so we instead sketch the quasimodes shifted up by the corresponding value of $\ell$. }
\label{figure}
\end{figure}

We also note here that the quasimodes we construct can be chosen to have trivial dependence on the coordinates parameterizing the compact dimensions. Note that this is \emph{not} the case for the non-decaying solutions constructed on three charge geometries in section \ref{section non decay}, which must have a nontrivial dependence on one of the compact dimensions. If we wish to specialise to solutions of the wave equation with trivial dependence on the compact directions, then we obtain the same results for both two and three charge microstates, i.e.\ we have a bound on the decay rate which is slower than logarithmic, again matching the numerical results of \cite{Eperon2016}.

When considering these results, we must bear in mind the following caveat: when using quasimodes to construct slowly decaying solutions, we do not actually \emph{construct} a solution which decays at a slower rate than logarithmic -- indeed, it may be the case that all solutions with suitable initial data decay faster\footnote{This can be contrasted with the construction of non-decaying solutions in the three charge case, in which we actually construct initial data for a solution whose local energy does not decay.}. However, we do show that no \emph{uniform} decay statement with a faster rate of decay can hold. To be more precise, we show that there are positive constants $C_m$ such that, for solutions $u$ to the wave equation arising from Schwartz initial data, we have
\begin{equation}
 \limsup_{t\rightarrow \infty} \sup_{u} \left( \frac{\log(2+t)}{\log\log(2+t)} \right)^m \frac{ ||u||_{H^1(U)} + ||\partial_t u||_{L^2(U)} }{\sqrt{E^{(m)}_0(u)}} \geq C_m
\end{equation}

Note also that we do not address the issue of whether all solutions with suitable initial data actually \emph{do} decay. If, as above, we restrict to waves with trivial dependence on the compact directions, then we can use the generic results of \cite{Moschidis2015} to prove the decay estimate
\begin{equation}
 || u ||_{H^1(U)} + || \partial_t u ||_{L^2(U)} \leq C \left( \log(2 + t) \right)^{-m} \sqrt{ E^{(m+1)}_0(u) }
\end{equation}
although we do not expect this to be sharp, in the sense that we expect a slightly faster rate of decay when estimating the solution in terms of a given number of derivatives of the initial data.

\vspace{5mm}

In summary, we perform a thorough analysis of the behaviour of solutions to the linear wave equation on supersymmetric microstate geometries. In both the two and three charge cases, we establish in section \ref{section boundedness} that the local energy of solutions arising from suitable initial data is bounded at all times. This is particularly surprising in the three charge case, given that these geometries possess an ergoregion and so might be thought to suffer from Friedman's ergosphere instability. On the other hand, following Friedman, on the three charge geometries we are able (in section \ref{section non decay}) to construct solutions whose local energy does not decay in time. Finally, in section \ref{section slow decay} we construct quasimodes on two charge microstate geometries, which we then use to show that no uniform decay statement can hold, except those with very slow (sub-logarithmic) decay rates.

\section{The geometries and their properties}

We will study both two and three charge microstate geometries. Here we describe these geometries and their respective metrics, and discuss some of their basic properties. See \cite{Maldacena2002} \cite{Balasubramanian2001}, \cite{Lunin:2002iz} \cite{Giusto2004} \cite{Giusto2004-2} \cite{Bena2006} \cite{Berglund2006} and \cite{Eperon2016} and the references therein for additional details.

\subsection{Three charge microstates}
The three charge microstate geometries are $\mathbb{R} \times \mathbb{S}^3 \times \mathbb{R}^2$, where points in $\mathbb{R}$ are given a coordinates $t$, points in $\mathbb{S}^3$ are given standard Hopf coordinates $(\theta, \phi, \psi)$, and points in $\mathbb{R}^2$ are given coordinates $(r, z)$, which are related to the standard polar coordinates $(r, \vartheta)$ by the identification $z = \frac{2\pi}{R_z} \vartheta$. This manifold is equipped with the metric
\begin{equation}
\label{equation three charge metric}
 \begin{split}
  g &= -\frac{1}{h}(\upd t^2 - \upd z^2) 
  + \frac{Q_p}{h f}(\upd t - \upd z)^2
  + hf\left( \frac{\upd r^2}{r^2 + (\tilde{\gamma}_1 + \tilde{\gamma}_2)^2\eta} + \upd \theta^2 \right) \\
  &\phantom{=} + h\left( r^2 + \tilde{\gamma_1}(\tilde{\gamma}_1 + \tilde{\gamma}_2)\eta - \frac{(\tilde{\gamma}_1^2 - \tilde{\gamma}_2^2)\eta Q_1 Q_2 \cos^2 \theta}{h^2 f^2} \right)\cos^2 \theta \upd \psi^2 \\
  &\phantom{=} + h\left( r^2 + \tilde{\gamma}_2(\tilde{\gamma}_1 + \tilde{\gamma}_2)\eta + \frac{(\tilde{\gamma}_1^2 - \tilde{\gamma}_2^2)\eta Q_1 Q_2 \sin^2 \theta}{h^2 f^2} \right)\sin^2 \theta \upd \phi^2 \\
  &\phantom{=} + \frac{Q_p (\tilde{\gamma}_1 + \tilde{\gamma}_2)^2 \eta^2}{h f}(\cos^2 \theta \upd \psi + \sin^2 \theta \upd \phi)^2 \\
  &\phantom{=} - 2\frac{\sqrt{Q_1 Q_2}}{h f}\left( \tilde{\gamma}_1 \cos^2 \theta \upd \psi + \tilde{\gamma}_2 \sin^2 \theta \upd \phi \right)(\upd t - \upd z) \\
  &\phantom{=} - 2\frac{(\tilde{\gamma}_1 + \tilde{\gamma}_2) \eta \sqrt{Q_1 Q_2}}{h f}\left( \cos^2 \theta \upd \psi + \sin^2 \theta \upd \phi \right) \upd z
 \end{split}
\end{equation}
where 
\begin{equation}
 \begin{split}
  \eta &= \frac{Q_1 Q_2}{Q_1 Q_2 + Q_1 Q_p + Q_2 Q_p} \\
  \tilde{\gamma}_1 &= -a\tilde{n} \\
  \tilde{\gamma}_2 &= a(\tilde{n}+1) \\
  Q_p &= a^2 \tilde{n}(\tilde{n}+1) \\
  f &= r^2 + (\tilde{\gamma}_1 + \tilde{\gamma}_2)\eta\left( \tilde{\gamma}_1 \sin^2 \theta + \tilde{\gamma}_2 \cos^2 \theta \right) \\
  h &= \sqrt{ \left( 1 + \frac{Q_1}{f}\right) \left( 1 + \frac{Q_2}{f} \right)} \\
  a &= \frac{\sqrt{Q_1 Q_2}}{R_z}
 \end{split}
\end{equation}
and the ranges of the coordinates are $t \in \mathbb{R}$, $0 \leq \theta \leq \frac{\pi}{2}$, $r \geq 0$, $0 \leq \phi$, $\psi \leq 2\pi$ and $0 \leq z \leq R_z$. 

The $z$ coordinate parametrises a ``Kaluza Klein'' circle of radius $R_z$, i.e.\ the coordinate values $z = 0$ and $z = R_z$ are identified, as indicated by the description of the manifold given above. Note that the $(r,z)$ plane asymptotically has the geometry of a cylinder with radius $R_z$, and not a flat plane. The coordinates $\theta$, $\phi$ and $\psi$ parametrize a $3$-sphere. For more details on the global structure of these spacetimes, and other similar spacetimes, see \cite{Gibbons2013}.

The quantities $Q_1$, $Q_2$ and $Q_p$ are the three ``charges'' of the spacetime: these are constants taking values in $\mathbb{R}_+ \setminus \{0\}$. It is useful to express the charge $Q_p$ in terms of the (non-negative) integer $\tilde{n}$ and another (non-negative) real number $a$, which is itself given by the square root of the products of the charges $Q_1$ and $Q_2$, divided by the period of the $z$ coordinate. Then the constants $\eta$, $\tilde{\gamma}_1$, $\tilde{\gamma}_2$ and the functions $f$ and $h$ are expressed in terms of these constants, together with (in the cases of $f$ and $h$) the coordinate functions $r$ and $\theta$.

The reader should note the following important facts regarding these manifolds:
\begin{itemize}
	\item These spacetimes are non-singular, in fact, the metric is everywhere smooth.
	\item The spacetimes are globally hyperbolic.
	\item There is a notion of ``null infinity'' for these spacetimes (see below), and with respect to this notion the spacetimes do not have a black hole region.
\end{itemize}

The six dimensional spacetime described by this metric is asymptotically Kaluza-Klein in the sense that
\begin{equation}
 g = -\upd t^2 + \upd z^2 + \upd r^2 + r^2 \slashed{g}_{\mathbb{S}^3} + \mathcal{O}(r^{-1})
\end{equation}
where $\slashed{g}_{\mathbb{S}^3}$ is the standard round metric on the unit $3$-sphere with (Hopf) coordinates $\theta, \phi, \psi$, and the norm of a tensor is defined relative to a basis of $1$-forms $\upd t$, $\upd z$ and $r e_A$, where $e_A$ are an orthonormal basis for the cotangent space of the unit 3-sphere $\mathbb{S}^3$. Together with additional fields, the metric given above provides a solution to the supergravity equations. Note that the metric is regular everywhere, including on the surface defined by $f = 0$, since near this surface $h \sim f^{-1}$. In addition, as $r \rightarrow 0$ the Kaluza-Klein circle shrinks to zero size while the $3$-sphere does not, however, coordinates can be found showing that the metric is regular \cite{Giusto2004}.

The solution possesses $4$ Killing vector fields as well as a ``hidden'' symmetry, which allows us to separate the wave equation (and the geodesic equation). The most important vector fields for our purposes (all of which are Killing) are 
\begin{equation}
 \begin{split}
  T &:= \left.\frac{\partial}{\partial t}\right|_{z, r, \theta, \phi, \psi} \\
  Z &:= \left.\frac{\partial}{\partial z}\right|_{t, r, \theta, \phi, \psi} \\
  \Psi &:= \left.\frac{\partial}{\partial \psi}\right|_{t, z, r, \theta, \phi} \\
  \Phi &:= \left.\frac{\partial}{\partial \phi}\right|_{t, z, r, \theta, \psi} \\
  V &:= T + Z
 \end{split}
\end{equation}
In particular, the last Killing vector field in the list above, $V$, is null everywhere and is future-directed. In contrast, the vector field $T$ is future directed and timelike at large $r$, and (when none of the charges vanish) spacelike at small $r$. This spacetime therefore has a genuine ergoregion associated with the vector field $T$. Indeed, we can compute
\begin{equation*}
	g(T,T)
	=
	-\frac{1}{h} \left( 1 - \frac{Q_p}{f} \right)
\end{equation*}
Since $h > 0$, the ergoregion is given by the region in which $f < Q_p$, i.e.\ it is the region
\begin{equation*}
	r^2
	<
	a^2 \left(
		\tilde{n} (\tilde{n} + 1)
		+ \eta \tilde{n} \sin^2 \theta
		- \eta (\tilde{n} + 1) \cos^2 \theta
	\right)
\end{equation*}
Clearly, if $a$, $\tilde{n}$ and $\eta$ are (strictly) positive then there is some region in which this condition holds (e.g.\ close to $\theta = \frac{\pi}{2}$).

A notion of ``evanescent ergosurface'' can also be introduced for this spacetime as in \cite{Eperon2016}, where it was defined as the submanifold on which $Z$ and $V$ are orthogonal, i.e.\ $g(Z, V) = 0$. This plays an important role when considering solutions of the wave equation which have trivial dependence on the $z$ coordinate, which we shall briefly outline here. Associated with the vector field $V$ is a non-negative ``energy'', but since $V$ is null this energy is degenerate. Specifically, it does not control derivatives of the solutions in the $V$ direction, although it controls derivatives in all the other directions. However, when $Z$ and $V$ are not orthogonal, derivatives in the $V$ direction can be expressed in terms of $Z$ derivatives and other derivatives which are controlled by the $V$ energy. Hence, the $V$ energy \emph{does} control all of the derivatives of a field which has trivial $Z$ dependence, except at the points where $Z$ and $V$ are orthogonal, at which the $V$ energy once again becomes degenerate. Hence, for these kinds of solutions, the submanifold defined by $g(Z, V) = 0$ plays the role of an evanescent ergosurface.

\subsection{Two charge microstates}
The manifold of a two charge microstate is identical to that of a three charge microstate: as before it can be viewed as $\mathbb{R} \times \mathbb{S}^3 \times \mathbb{R}^2$ with coordinates $t$ on $\mathbb{R}$, $(\theta, \phi, \psi)$ (Hopf coordinates) on $\mathbb{S}^3$, and $(r, z)$ on $\mathbb{R}^2$, which are related to polar coordinates in the same was as for a three charge microstate.

The metric of two charge microstates can be obtained from the metric for three charge microstates by setting $Q_p = 0$, which in turn means that $\tilde{n} = \tilde{\gamma}_1 = 0$ and $\eta = 1$. We summarize the important differences between the two and three charge microstate geometries below.

The $T$ Killing vector field is never spacelike in the two charge microstate geometry, in contrast to the three charge microstate geometry. However, it does become null on the submanifold defined by $r = 0$, $\theta = \frac{\pi}{2}$. Thus, unlike the three charge geometry, the two charge microstate geometry does not have an ergoregion but only an ``evanescent ergosurface'', which is reminiscent to the boundary of an ergoregion.

As $r \rightarrow 0$ and $\theta \rightarrow \frac{\pi}{2}$, the Kaluza-Klein circle smoothly pinches off to zero size, and we find that the submanifold $r = 0$, $\theta = \frac{\pi}{2}$, $t = \text{const.}$ (on which $T$ is null) has dimension $1$. In fact, points on this submanifold are uniquely specified by their $\phi$ coordinate, as expected from the fact that $(\theta, \psi, \phi)$ are Hopf coordinates (the level sets of $\theta$ on $\mathbb{S}^3$ are tori for $\theta \in (0, \frac{\pi}{2})$, parametrised by $(\phi, \psi)$, while the levels sets $\theta = 0$ and $\theta = \frac{\pi}{2}$ are circles parametrised by $\psi$ and $\phi$ respectively).

For reference, we provide the metric of the two charge microstate geometries we consider below:
\begin{equation}
\label{equation two charge metric}
 \begin{split}
  g &= -\frac{1}{h}(\upd t^2 - \upd z^2)
    + hf\left( \upd \theta^2 + \frac{\upd r^2}{r^2 + a^2}\right) 
    - \frac{2a \sqrt{Q_1 Q_2}}{hf}\left( \cos^2 \theta \upd z \upd \psi + \sin^2 \theta \upd t \upd \phi \right) \\
    &\phantom{=}+ h\left( \left( r^2 + \frac{a^2 Q_1 Q_2 \cos^2 \theta}{h^2 f^2}\right) \cos^2 \theta \upd \psi^2
    + \left( r^2 + a^2 - \frac{a^2 Q_1 Q_2 \sin^2 \theta}{h^2 f^2} \right)\sin^2 \theta \upd \phi^2 \right) 
 \end{split}
\end{equation}
where as before
\begin{equation}
 a := \frac{\sqrt{Q_1 Q_2}}{R_z}
\end{equation}

\section{Preliminary calculations and notation}
First we shall need several preliminary calculations which serve to set up notation and to prove some basic statements.

\begin{definition}[Notation]
 We shall use the notation
\begin{equation*}
 a \lesssim b
\end{equation*}
to indicate that there is some positive constant $C > 0$, independent of all parameters that are varying in our set up, such that
\begin{equation*}
 a \leq C b
\end{equation*}
Similarly, we shall sometimes use the notation $a \gtrsim b$. Finally, we use the notation
\begin{equation*}
 a \sim b
\end{equation*}
to indicate that there are positive constants $c$, $C > 0$ such that
\begin{equation*}
 cb \leq a \leq Cb
\end{equation*}
Note that the charges $Q_1$, $Q_2$, $Q_p$ and the parameters $\tilde{\gamma}_1$, $\tilde{\gamma}_2$ will be considered \emph{fixed} parameters during our calculations, so that, for example, $a \lesssim b$ means that there is some constant $C$, which may depend on $Q_1$, $Q_2$, $Q_p$, $\tilde{\gamma}_1$ and $\tilde{\gamma}_2$ but which is independent of all other parameters, such that $a \leq C b$.

We also use ``musical notation'': for any pair of covectors $\mu$ we define the vector $\mu^\sharp$ by
\begin{equation}
 \nu( \mu^\sharp ) = g^{-1}(\nu, \mu)
\end{equation}
for \emph{any} covector $\nu$. Similarly, given any vector $V$ we define the covector $V^\flat$ by
\begin{equation}
 V^\flat (X) = g(V, X)
\end{equation}
for \emph{any} vector $X$.

\end{definition}

\begin{definition}[The energy momentum tensor]
 We define the energy momentum tensor associated to a function $u$ as follows:
\begin{equation}
 Q_{\mu \nu}[u] := (\partial_\mu u)(\partial_\nu u) - \frac{1}{2}g_{\mu \nu}\left( (g^{-1})^{\alpha \beta}(\partial_\alpha u)(\partial_\beta u) \right)
\end{equation}
\end{definition}

\begin{definition}[Energy currents]
 Given a vector field $X$ and a function $u$ we define the associated energy current:
\begin{equation}
 {^{(X)}J}^\mu[u] := X^\nu Q_{\nu}^{\phantom{\nu}\mu}[u]
\end{equation}
 We shall sometimes refer to the vector field $X$ as a ``multiplier''.
\end{definition}

\begin{definition}[Deformation tensors]
 Given a vector field $X$ we define the associated deformation tensor
\begin{equation}
 {^{(X)}\pi}_{\mu \nu} := (\mathcal{L}_X g)_{\mu \nu} = \nabla_\mu X_\nu + \nabla_\nu X_\mu
\end{equation}
where $\nabla$ is the Levi-Civita connection associated with $g$. In particular, if $X$ is a Killing vector of $g$ then ${^{(X)}\pi} = 0$.
\end{definition}

We have the following classical energy identity, which is a consequence of the divergence theorem:
\begin{proposition}[The energy identity]
\label{proposition energy identity}
 Let $\mathcal{U}$ be a compact open set with smooth boundary $\partial \mathcal{U}$, and let $u$ and $X$ be smooth. Then
\begin{equation}
 \int_{\partial \mathcal{U}} \imath_{{^{(X)}J}[u]} \dVol_g = \int_{\mathcal{U}} \left( \frac{1}{2} {^{(X)}\pi}^{\mu \nu} Q_{\mu \nu}[u] + (Xu)\Box_g u \right) \dVol_g
\end{equation}
where $\dVol_g$ denotes the volume form associated with $g$, and $\imath$ denotes the interior product. Moreover, the same statement holds if $\mathcal{U}$ is not compact but $u$ decays sufficiently rapidly and $|X|$ is bounded.
\end{proposition}
In particular, proposition \ref{proposition energy identity} means that, if $u$ solves the wave equation and $X$ is a Killing vector field of $g$, then
\begin{equation*}
 \int_{\partial \mathcal{M}} \imath_{{^{(X)}J}[u]} \dVol_g = 0
\end{equation*}

We shall introduce notation for several regions of the spacetime manifold:
\begin{definition}[Submanifolds of the microstate geometries]
We use the notation $\mathcal{M}$ for the manifold associated with either the two or three charge microstate geometry.

In both the two and three charge microstate geometries we define the hypersurfaces of constant $t$:
\begin{equation}
 \Sigma_{t_1} := \{x \in \mathcal{M} \ \big| \ t(x) = t_1\}
\end{equation}
as well as the open spacetime region $\mathcal{M}_{t_0}^{t_1} \subset M$
\begin{equation}
 \mathcal{M}_{t_0}^{t_1} := \{ x \in \mathcal{M} \ \big| t_0 < t(x) < t_1 \}
\end{equation}

We also define the ``evanescent ergosurface'' as the submanifold $\mathcal{S}_t \subset \Sigma_t$ defined by
\begin{equation}
 \mathcal{S}_t := \{ x \in \Sigma_t \ \big| \ f\left( r(x), \theta(x) \right) = 0 \}
\end{equation}
Note that in the case of the three charge microstate geometries the submanifold $\mathcal{S}_t$ is a genuine hypersurface within $\Sigma_t$, i.e.\ a co-dimension one submanifold of $\Sigma_t$, while in the case of the two charge microstate geometries, the submanifold $\mathcal{S}_t$ is the co-dimension four (i.e.\ one dimensional) submanifold given by
\begin{equation*}
 \mathcal{S}_t = \left\{ x \in \Sigma_t \ \big| \ r = 0 \ , \ \theta = \frac{\pi}{2} \right\}
\end{equation*}

In the two charge microstate geometries, we define the open region (as a subset of $\Sigma_t$) containing the evanescent ergosurface:
\begin{equation}
 \tilde{\mathcal{S}}_t^\epsilon := \left\{ x \in \Sigma_t \ \big| \ r < \epsilon \ , \ \frac{\pi}{2} - \theta < \epsilon \right\}
\end{equation}

In the three charge microstate geometries, we can define the ergoregion as the open region given by:
\begin{equation}
 \mathcal{E}_t := \left\{ x \in \Sigma_t \ \big| \ g(T, T) > 0 \right\}
\end{equation}
Similarly, we define a slightly enlarged region containing the ergoregion as follows:
\begin{equation}
 \tilde{\mathcal{E}}_t^\epsilon := \left\{ x \in \Sigma_t \ \big| \ g(T, T) > -\epsilon \right\}
\end{equation}

\end{definition}

We shall also need the following properties of the microstate geometry metrics, which can be found in (\cite{Maldacena2002} \cite{Balasubramanian2001}, \cite{Lunin:2002iz} \cite{Giusto2004} \cite{Giusto2004-2} \cite{Bena2006} \cite{Berglund2006}):
\begin{proposition}[The volume form]
 On both the three charge and two charge microstate geometries, the volume form induced by the metric is given by
\begin{equation}
 \dVol_g = -r h f \sin\theta \cos \theta \ \upd t \wedge \upd z \wedge \upd r \wedge \upd \theta \wedge \upd \phi \wedge \upd \psi
\end{equation}
\end{proposition}

\begin{proposition}[The hypersurfaces $\Sigma_t$ are (uniformly) spacelike]
 On both the three charge and two charge microstate geometries, the hypersurfaces $\Sigma_t$ are uniformly spacelike. Indeed, we have (see the comments under equation (5.4) in \cite{Giusto2004-2})
\begin{equation}
 g^{-1}(\upd t , \upd t) = -\frac{1}{hf} \left( f + Q_1 + Q_2 + Q_p + \frac{ Q_1 Q_2 + Q_1 Q_p + Q_2 Q_p}{r^2 + (\tilde{\gamma}_1 + \tilde{\gamma}_2)^2 \eta } \right)
\end{equation}
which is bounded both above and below by some negative constants, depending on the constants $Q_1$, $Q_2$, $Q_p$, $\tilde{\gamma}_1$ and $\tilde{\gamma}_2$. Hence the hypersurfaces of constant $t$ are uniformly spacelike.
\end{proposition}

We can therefore make the following definition:
\begin{definition}[The vector field $n$]
 We define the vector $n$ as the unit (timelike) future directed normal to the hypersurface $\Sigma_t$. Note, from the above, that
 \begin{equation*}
 	n = -C (\upd t)^\sharp
 \end{equation*}
 where the function $C$ is uniformly bounded away from $0$ and $\infty$.
\end{definition}

\begin{definition}[The metric and volume form on $\Sigma_t$]
 We denote by $\underline{g}$ the metric induced by $g$ on the hypersurfaces $\Sigma_{\tau}$, and similarly we write $\dVol_{\underline{g}}$ for the volume form induced on the surfaces $\Sigma_\tau$. Since these hypersurfaces are \emph{uniformly} spacelike, we find that
 \begin{equation}
  \dVol_{\underline{g}} \sim r h f \sin\theta \cos \theta \ \upd z \wedge \upd r \wedge \upd \theta \wedge \upd \phi \wedge \upd \psi
 \end{equation}
\end{definition}

\begin{definition}[The nondegenerate energy]
 We define the non-degenerate energy of a function $u$ as follows: let $\{X_1, X_2,\ldots, X_5\}$ be an orthonormal basis for the tangent space of $\Sigma_t$ at the point $p \in \Sigma_t$. Then we define the non-degenerate energy at the point $p$ as follows:
\begin{equation}
 |\partial u|^2 := |T u|^2 + \sum_{A = 1}^5 |X_A u|^2
\end{equation}
Note that, since $T$ is transverse to the hypersurface $\Sigma_t$, the set $\{ T, X_1, \ldots, X_5 \}$ spans the tangent space of $\mathcal{M}$ at the point $p$.
\end{definition}

\begin{definition}[Higher order nondegenerate energies]
 We define the higher order non-degenerate energy on the surface $\Sigma_t$ as follows:
\begin{equation}
 E_t^{(m)} := \sum_{|\alpha| \leq m} \int_{\Sigma_t} |\partial \partial^{\alpha} u|^2 \dVol_{\underline{g}}
\end{equation}
for multi-indices $\alpha$, where (as above) $\partial$ may refer to any derivative in the set $\{T, X_1, \ldots, X_5 \}$.
\end{definition}

\begin{definition}[``Good'' derivatives]
 We also define a subset of the derivatives appearing in the non-degenerate energy to be the ``good'' derivatives. To be specific, we define $N$ to be the vector field (tangent to $\Sigma_t$) obtained by orthogonally projecting $T$ onto the surface $\Sigma_t$, that is,
\begin{equation}
 N := T + g(T, n)n
\end{equation}
 away from regions in which $n$ is parallel to $T$, we can define the last vector in the orthonormal basis $\{ X_1, \ldots, X_5\}$ to be parallel to $N$, i.e.\
 \begin{equation}
  X_5 := \left( g(N, N) \right)^{-\frac{1}{2}} N \\
 \end{equation}
 Indeed, we can choose the vector $X_5$ such that it is always parallel to $N$, so that
\begin{equation}
 N = (g(N, N))^{\frac{1}{2}} X_5
\end{equation}

 This allows us to define the ``good derivatives'' at any point $x \in \mathcal{M}$:
 \begin{equation}
  |\bar{\partial} u|^2 := |T u|^2 + \sum_{A = 1}^4 |X_A u|^2
 \end{equation}
 In other words, the good derivatives $\bar{\partial}u$ exclude the derivative in the $N$ direction.
\end{definition}

\begin{proposition}[The $T$-energy current]
\label{proposition T energy current}
	The energy current appearing in the energy identity in proposition \ref{proposition energy identity}, with the choice $X = T$, is given by
	\begin{equation}
	\begin{split}
		\left(\imath_{{^{(T)}J}[u]} \dVol_g \right)\big|_{\Sigma_t} 
		&= 
		\frac{1}{2}\bigg( -g(T, n)^{-1}(T u)^2 - g(T, n)\sum_{A = 1}^4 (X_A u)^2 \\
		&\phantom{= \frac{1}{2}\bigg(}
		 - g(T, n)^{-1}\left( (g(T, n))^2 - g(N, N) \right)(X_5 u)^2 \bigg) 
		 \dVol_{\underline{g}}
	\end{split}
	\end{equation} 
\end{proposition}

\begin{proof}
	Apply the energy estimate of proposition \ref{proposition energy identity} with the multiplier $T$ to the function $u$, on the spacetime region $\mathcal{M}_0^t$. We find that, 
\begin{equation}
 \int_{\Sigma_0} \imath_{{^{(T)}J}[u]} \dVol_g = \int_{\Sigma_t} \imath_{{^{(T)}J}[\phi]} \dVol_g
\end{equation}
Recalling that $n$ is the unit future-directed normal to $\Sigma_t$, then we find, restricting the interior product of the energy current and the volume form to the hypersurface $\Sigma_t$,
\begin{equation}
 \begin{split}
  \left(\imath_{{^{(T)}J}[u]} \dVol_g \right)\big|_{\Sigma_t} &= \left( (n u)(T u) - \frac{1}{2}g(T, n)\cdot g^{-1}(\upd u ,\upd u) \right)\dVol_{\underline{g}} \\
  &= \left( (n u)(T u) - \frac{1}{2}g(T, n)\left( -(n u)^2 + \sum_{A = 1}^5 (X_A u)^2 \right) \right)\dVol_{\underline{g}}
 \end{split}
\end{equation}
which, after a bit of algebra, gives the expression in the proposition.
\end{proof}

\section{Avoiding the ``Friedman instability'': uniform boundedness for solutions to the wave equation}
\label{section boundedness}
In \cite{Friedman1978} an instability associated with spacetimes possessing an ergoregion but lacking an event horizon was proposed. To be more precise, it was shown that on such backgrounds, the energy of solutions to the scalar wave equation (and to the Maxwell equations) cannot decay within the ergoregion, and a heuristic argument was given to suggest that such solutions actually grow. Very recently, \cite{Moschidis:2016zjy} has provided a rigorous proof of this growth, under certain additional but still very general assumptions.

The two charge microstate geometries do not possess an ergoregion, and are therefore immune to even the heuristic arguments for instability of \cite{Friedman1978}. Nevertheless, the absence of a global, timelike Killing vector field means that it is not straightforward to show that solutions of the wave equation remain uniformly bounded in terms of the initial data, and \emph{a priori} it is conceivable that some remnant of the ergosphere instability might lead to an instability of geometries with an evanescent ergosurface. Nevertheless, we are able to obtain a boundedness statement on these geometries, with a ``loss of derivatives'', i.e.\ we can bound the energy of solutions in the future by an initial ``higher order'' energy, involving higher derivatives of the initial data. We note here that a statement of this kind is unlikely to prove useful in any kind of nonlinear application, although it can help us to understand the nature of \emph{linear} waves, such as those studied in this paper.

In contrast, the three charge microstate geometries \emph{do} possess an ergoregion, and thus might be expected to be unstable due to the ``ergoregion instability of \cite{Friedman1978}, \cite{Moschidis:2016zjy}. However, these geometries do not satisfy all of the conditions required in \cite{Moschidis:2016zjy}; most importantly, they are not asymptotically flat, but are instead asymptotically Kaluza-Klein. In fact, once again we are able to prove a uniform boundedness statement, although, similarly to the two charge case, we must ``lose derivatives''. A key part of this proof relies on the presence of the globally null Killing vector field $V$.

Another key part of the proof of uniform boundedness, both in two charge and three charge microstate geometries, is a version of Hardy's inequality, which we prove below:

\begin{lemma}
\label{lemma Hardy}
 Let $\mathcal{M}$ be either a two or a three charge microstate geometry, and let $u$ be a smooth function on $\Sigma_t$ such that
\begin{equation}
 \lim_{r \rightarrow \infty} r^2 |u|^2 = 0
\end{equation}
Define the vector field (which is tangent to $\Sigma_t$)
\begin{equation}
\label{equation definition R}
 R := \left. \frac{\partial}{\partial r}\right|_{t,z,\theta,\phi,\psi}
\end{equation}
then we have
\begin{equation}
 \int_{\Sigma_t} \frac{|u|^2}{(1+r)^2} \dVol_{\underline{g}} \lesssim \int_{\Sigma_t} |R u|^2 \frac{r^2}{(1+r)^2} \dVol_{\underline{g}}
\end{equation}

\end{lemma}

\begin{proof}
 Note that
\begin{equation*}
 r h f \sim r(1+r)^2
\end{equation*}
So we have
\begin{equation*}
 \begin{split}
  \int_{\Sigma_t} \frac{|u|^2}{(1+r)^2} \dVol_{\underline{g}} &\lesssim \int_{\Sigma_t} |u|^2 r \sin\theta \cos\theta \ \upd z \wedge \upd r \wedge \upd \theta \wedge \upd \phi \wedge \upd \psi \\
  &\lesssim \int_{\Sigma_t} |u|^2 (R r^2) \sin\theta \cos\theta \ \upd z \wedge \upd r \wedge \upd \theta \wedge \upd \phi \wedge \upd \psi \\
 \end{split}
\end{equation*}
Integrating by parts in the $r$ direction, using the fact that $u$ is smooth and decays suitably at infinity, we find that, for any $\delta > 0$,
\begin{equation*}
 \begin{split}
  \int_{\Sigma_t} \frac{|u|^2}{(1+r)^2} \dVol_{\underline{g}} &\lesssim \int_{\Sigma_t} \left( \delta |u|^2 r + \delta^{-1}|R u|^2 r^3 \right) \sin\theta \cos\theta \ \upd z \wedge \upd r \wedge \upd \theta \wedge \upd \phi \wedge \upd \psi \\
  &\lesssim \delta \int_{\Sigma_t} \frac{|u|^2}{(1+r)^2} \dVol_{\underline{g}} + \delta^{-1} \int_{\Sigma_t} |R u|^2 \frac{r^2}{(1+r)^2} \dVol_{\underline{g}}
 \end{split}
\end{equation*}
Taking $\delta$ sufficiently small, we can absorb the first term by the left hand side, proving the lemma.

\end{proof}

We will also need the following result, which allows us to compute the equation satisfied by the commuted field, and which follows from a simple calculation:
\begin{proposition}
\label{proposition commute}
Let $u$ be a smooth function and let $V$ be a smooth vector field on $\mathcal{M}$. Then we have
\begin{equation}
 \Box_g ( Vu ) = V^\alpha \nabla_\alpha ( \Box_g u) + {^{(V)}\pi}^{\mu \nu}\nabla_\mu \nabla_\nu u + \left( \nabla_\mu  \left({^{(V)}\pi}^{\mu \alpha}\right) - \frac{1}{2} \nabla^\alpha \left({^{(V)}\pi}_\mu^{\phantom{\mu}\mu}\right) \right)\nabla_\alpha u
\end{equation}
In particular, if $\Box_g u = 0$ and $V$ is a Killing vector field for the metric $g$, then $\Box_g (V u) = 0$
\end{proposition}

\subsection{Uniform boundedness on two charge microstate geometries}
In this section we will prove uniform boundedness for solutions of the wave equation
\begin{equation}
 \Box_g u = 0
\end{equation}
on two charge microstate geometries \eqref{equation two charge metric}.

We begin by applying the energy estimate with the multiplier $T$, which allows us to prove the following:
\begin{proposition}[The $T$-energy estimate on two charge microstates]
\label{proposition T energy two charge}
Let $u$ solve the wave equation $\Box_g u = 0$ on a two charge microstate geometry, with metric \eqref{equation two charge metric}. Then for any $t \geq 0$ and any $\epsilon > 0$ we have the following \emph{degenerate} energy estimate:
\begin{equation}
\label{equation T energy two charge 1}
 \int_{\Sigma_t \setminus \tilde{\mathcal{S}}_t^\epsilon}|\partial u|^2 \dVol_{\underline{g}} + \int_{\tilde{\mathcal{S}}_t^\epsilon} |\bar{\partial} u|^2 \dVol_{\underline{g}} \lesssim \int_{\Sigma_0} |\partial u|^2 \dVol_{\underline{g}}
\end{equation}
In addition, there is a positive constant $C(\epsilon) > 0$ such that
\begin{equation}
\label{equation T energy two charge 2}
 \int_{\Sigma_t} |\partial u|^2 \dVol_g \leq C(\epsilon)\left( \int_{\tilde{\mathcal{S}}_t^\epsilon}|N u|^2 + \int_{\Sigma_0}|\partial u|^2 \dVol_g \right)
\end{equation}
\end{proposition}

\begin{proof}
 
We begin with proposition \ref{proposition T energy current}. Note that
\begin{equation}
 n = -\left( -g^{-1}(\upd t, \upd t) \right)^{\frac{1}{2}} (\upd t)^\sharp
\end{equation}
and so
\begin{equation}
 g(T, n) = -\left( -g^{-1}(\upd t, \upd t) \right)^{\frac{1}{2}} \sim -1
\end{equation}
Moreover, in the region $\Sigma_t \setminus \tilde{\mathcal{S}}_t^\epsilon$ we have $g(T, T) \leq -\epsilon$. Since $T = -g(T, n)n + N$, we have
\begin{equation}
 (g(T, n))^2 - g(N, N) \geq \epsilon \quad \text{in the region } \Sigma_t \setminus \tilde{\mathcal{S}}_t^\epsilon
\end{equation}
In particular we find that there is some positive constant $c(\epsilon) > 0$ such that
\begin{equation}
 \int_{\Sigma_t \setminus \tilde{\mathcal{S}}_t^\epsilon} \imath_{{^{(T)}J}[u]} \dVol_g \geq c(\epsilon) \int_{\Sigma_t \setminus \tilde{\mathcal{S}}_t^\epsilon} |\partial u|^2 \dVol_{\bar{g}} 
\end{equation}
In fact, we have that
\begin{equation}
 (g(T, n))^2 - g(N, N) \geq 0
\end{equation}
with equality on and only on $\mathcal{S}$. In order to prove the proposition, we only need to check that $N \neq 0$ in the region $\tilde{\mathcal{S}}_t^\epsilon$. But in this region, we have
\begin{equation}
 (g(T, n))^2 - g(N, N) \leq \epsilon
\end{equation}
Since $g(T, n) \sim 1$, we find that $g(N, N)$ is bounded away from zero as long as $\epsilon$ is sufficiently small.
\end{proof}

The previous proposition proves boundedness of the \emph{degenerate} energy, but we wish to conclude boundedness of the \emph{non-degenerate} energy, including the derivatives in the $N$ direction on the submanifold $\mathcal{S}$. In order to do this, we will first have to \emph{commute} the wave equation with a suitable operator, and we will then need to make use of the Hardy inequality of lemma \ref{lemma Hardy}.

Note that $\Phi$ is a Killing vector field for the two charge microstate geometry (and also for the three charge microstate geometry), and so from proposition \ref{proposition commute} we find that the field $(\Phi u)$ satisfies the wave equation $\Box_g (\Phi u) = 0$.

We now aim to show that, in the region $\tilde{\mathcal{S}}_t^\epsilon$, for sufficiently small $\epsilon$, we can express the vector field $N$ in terms of the $\Phi$ and the ``good'' derivatives:

\begin{proposition}[Expressing $N$ in terms of $\Phi$]
\label{proposition N in terms of Phi}
 For all sufficiently small $\epsilon$, we have
 \begin{equation}
  N \in \spn\{ \Phi, X_1, X_2, X_3, X_4 \}
 \end{equation}
 in the region $\tilde{\mathcal{S}}_t^\epsilon$, and moreover $N$ can be expressed as
 \begin{equation}
  N = N^\Phi \Phi + \sum_{A = 1}^4 N^A X_A
 \end{equation}
 where there exists some constant $C(\epsilon) > 0$ such that 
 \begin{equation}
  |N^\Phi| + \sum_{A = 1}^4 |N^A| \lesssim C(\epsilon)
 \end{equation}
 in the region $\tilde{\mathcal{S}}_t^\epsilon$.
\end{proposition}
\begin{proof}
 Decomposing $\Phi$ in the basis $\{ n, X_1, \ldots, X_5\}$ we write
 \begin{equation*}
  \Phi = \Phi^0 n + \sum_{A = 1}^5 \Phi^A X_A
 \end{equation*}
 Taking the inner product with $n$, recalling that $n \propto \upd t^\sharp$ and $\Phi(t) = 0$ we find that $\Phi^0 = 0$. On the other hand, we have that
 \begin{equation*}
  \Phi^5 = g(\Phi, X_5) = g(N, N)^{-\frac{1}{2}} g(\Phi, N)
 \end{equation*}
 where we recall that, in the region $\tilde{\mathcal{S}}_t^\epsilon$, if $\epsilon$ is sufficiently small then there is some constant $C(\epsilon) > 0$ such that
 \begin{equation*}
  g(N, N) \sim C(\epsilon)
 \end{equation*}
 Now, we can write
 \begin{equation*}
  \Phi = \sum_{A = 1}^4 \Phi^A X_A + g(N, N)^{-1} g(\Phi, N) N
 \end{equation*}
 So we need a lower bound on $|g(\Phi, N)|$ in order to prove the proposition.
 
 Now, recalling the definition of $N$, we find that
 \begin{equation*}
  \begin{split}
   g(\Phi, N) &= g(\Phi, T) + g(T, n)g(\Phi, n) \\
   &= g(\Phi, T) \\
   &= -\frac{a \sqrt{Q_1 Q_2}}{hf} \sin^2 \theta
  \end{split}
 \end{equation*}
 since $g(\Phi, n) = 0$. However, in the region $\tilde{\mathcal{S}}_t^\epsilon$ we have $|\theta - \frac{\pi}{2}| < \epsilon$ so that $\sin \theta$ is bounded away from zero. Additionally, we have $hf \rightarrow \sqrt{Q_1 Q_2}$ as $r \rightarrow 0$, $\theta \rightarrow \frac{\pi}{2}$ so that $|hf|$ is bounded away from zero in the region in question.
 
 In order to finish the proof of the proposition, we note that
 \begin{equation}
  g(\Phi, \Phi) = \left( r^2 + a^2 - \frac{a^2 Q_1 Q_2 \sin^2 \theta}{h^2 f^2} \right)\sin^2 \theta
 \end{equation}
which is bounded by some constant (depending on $\epsilon$) in the region $\tilde{\mathcal{S}}_t^\epsilon$.
\end{proof}

We require one more proposition, establishing that the vector field $R$ is in the span of the ``good'' vector fields in any region where $|N|$ is bounded away from zero.
\begin{proposition}
\label{proposition R good derivative}
 Let $x \in \mathcal{M}$ be a point such that $g(N,N) \neq 0$ at the point $x$. Then we have
\begin{equation}
 R \in \spn\{X_1, X_2, X_3, X_4\} \text{at } x
\end{equation}
where we recall the definition of the vector field $R$ given in equation \eqref{equation definition R}.

Moreover, we have
\begin{equation}
 R = \sum_{A = 1}^4 R^A X_A
\end{equation}
where
\begin{equation}
 \sum_{A = 1}^4 |R^A| \lesssim 1
\end{equation}
\end{proposition}
\begin{proof}
 Decomposing in the orthonormal basis $\{n, X_1, \ldots, X_5\}$ we can write
\begin{equation*}
 R = R^0 n + \sum_{A = 1}^5 R^A X_A
\end{equation*}
 We immediately conclude that $R^0 = 0$ since $R(t) = 0$ and $n \propto \upd t$. Thus we only need to show that $R^5 = 0$.

 Since $g(N, N) \neq 0$ at $x$, we can write
\begin{equation*}
 X_5 = g(N,N)^{-\frac{1}{2}} N
\end{equation*}
 and so we find that
\begin{equation*}
 \begin{split}
  R^5 &= g(N, N)^{-\frac{1}{2}} g(R, N) \\
  &= g(N,N)^{-\frac{1}{2}} g(R, T)
 \end{split}
\end{equation*}
where the last line follows from the fact that $N = T + g(T, n)n$. But in both the two and three charge microstate geometries, we observe directly from equations \eqref{equation two charge metric} and \eqref{equation three charge metric} that $g(R, T) = 0$. Finally, to conclude the last part of the proposition we note that $g(R, R) \sim 1$.

\end{proof}

We are now ready to prove the uniform boundedness statement:

\begin{theorem}[Uniform boundedness of solutions to the wave equation on two charge microstate geometries]
\label{theorem boundedness two charge}
Let $u$ be a solution to the wave equation $\Box_g u = 0$ on a two charge microstate geometry, with compactly supported initial data on the initial hypersurface $\Sigma_0$. Then, for all $t \geq 0$ we have
\begin{equation}
 \int_{\Sigma_t} |\partial u|^2 \dVol_{\underline{g}} \lesssim \int_{\Sigma_0} \left( |\partial u|^2 + |\partial \Phi u|^2 \right)\dVol_{\underline{g}} 
\end{equation}
 
\end{theorem}

\begin{proof}
We begin by applying the $T$-energy estimate of proposition \ref{proposition T energy two charge} to both the field $u$ and the field $\Phi u$. We use the form of the $T$-energy inequality given in equation \eqref{equation T energy two charge 2} for the field $u$, and the form given in equation \eqref{equation T energy two charge 1} for both the field $u$ and the commuted field $\Phi u$. We then sum the resulting inequalities, multiplying both inequalities which are of the form \eqref{equation T energy two charge 1} by some large constant $C$ (to be fixed later). We obtain
\begin{equation}
\label{equation boundedness two charge v1}
 \begin{split}
  &\int_{\Sigma_t \setminus \tilde{\mathcal{S}}_t^\epsilon} \left( (1 + C)|\partial u|^2 + C|\partial \Phi u|^2 \right) \dVol_{\underline{g}} + \int_{ \tilde{\mathcal{S}}_t^\epsilon} \left( |\partial u|^2 + C|\bar{\partial} u|^2 + C|\bar{\partial} \Phi u|^2 \right) \dVol_{\underline{g}} \\
  &\lesssim \int_{\tilde{\mathcal{S}}_t^\epsilon} | N u|^2 \dVol_{\underline{g}} + \int_{\Sigma_0} \left( |\partial u|^2 + C |\partial \Phi u|^2 \right) \dVol_{\underline{g}}
 \end{split}
\end{equation}
Now, we use the fact that $R$ is in the span of the ``good derivatives'' in the region $\tilde{\mathcal{S}}_t^\epsilon$ as shown in proposition \ref{proposition R good derivative},  to show
\begin{equation*}
 \int_{ \tilde{\mathcal{S}}_t^\epsilon} \left( |\bar{\partial} \Phi u|^2 \right) \dVol_{\underline{g}} \gtrsim \int_{ \tilde{\mathcal{S}}_t^\epsilon} \left( |R \Phi u|^2 \right) \dVol_{\underline{g}} \\
\end{equation*}

Since $u$ has compact support on the initial data surface $\Sigma_0$, we can use the domain of dependence property of solutions to the wave equation to show that
\begin{equation*}
 \lim_{r \rightarrow \infty} r^2|u|^2 = 0
\end{equation*}
at any time $t \geq 0$. We can therefore appeal to the Hardy inequality of lemma \ref{lemma Hardy} and find that, for all sufficiently small $\epsilon$ there is some constant $C(\epsilon) > 0$ such that
\begin{equation*}
 \begin{split}
  \int_{\mathcal{S}_t^\epsilon} |\Phi u|^2 \dVol_{\underline{g}} &\leq C(\epsilon) \int_{\Sigma_t} |\Phi u|^2 \frac{1}{(1 + r)^2} \dVol_{\underline{g}} \\
  &\leq C(\epsilon) \int_{\Sigma_t} |R \Phi u|^2 \dVol_{\underline{g}} \\
  &\leq C(\epsilon) \int_{\Sigma_t} |\bar{\partial} \Phi u|^2 \dVol_{\underline{g}}
 \end{split}
\end{equation*}
where we have made use of the fact that $r$ is bounded above in the region $\tilde{\mathcal{S}}_t^\epsilon$
In particular, returning to equation \eqref{equation boundedness two charge v1} we can show
\begin{equation}
\label{equation boundedness two charge v2}
 \begin{split}
  \int_{\Sigma_t} |\partial u|^2 \dVol_{\underline{g}} + \int_{\tilde{\mathcal{S}}_t^\epsilon} C\left( |\Phi u|^2 + |\bar{\partial} u|^2 \right) \dVol_{\underline{g}}
  \lesssim \int_{\tilde{\mathcal{S}}_t^\epsilon} | N u|^2 \dVol_{\underline{g}} + \int_{\Sigma_0} \left( |\partial u|^2 + C |\partial \Phi u|^2 \right) \dVol_{\underline{g}}
 \end{split}
\end{equation}
Now, in the region $\tilde{\mathcal{S}}_t^\epsilon$, we have that $N \in \spn\{\Phi, X_1, X_2, X_3, X_4\}$ according to proposition \ref{proposition N in terms of Phi}. Thus, if we take the constant $C$ to be sufficiently large, then we can absorb the first term on the right hand side of equation \eqref{equation boundedness two charge v2} by the left hand side, finishing the proof of the theorem.
 
\end{proof}

\subsection{Uniform boundedness on three charge microstate geometries}
In this subsection we will prove uniform boundedness of waves on three charge microstate geometries. We will roughly follow the pattern of the proof of boundedness in two charge microstate geometries, first applying the $T$-energy estimate and then attempting to control the terms with the wrong sign by commuting and making use of the Hardy inequality. The main difference is that there is a genuine ergoregion in the three charge microstates, meaning that, rather than being degenerate, the $T$-energy can actually become negative in the ergoregion. As such, we must also make use of the globally null Killing vector field $V$ in order to obtain a degenerate energy estimate.

We begin by applying the energy estimate with the multiplier $T$, which allows us to prove the following:
\begin{proposition}[The $T$-energy estimate on three charge microstates]
\label{proposition T energy three charge}
Let $u$ solve the wave equation $\Box_g u = 0$ on a three charge microstate geometry, with metric \eqref{equation three charge metric}. Then for any $t \geq 0$ and any $\epsilon > 0$ we have the following energy estimate (with indefinite sign): there is a positive constant $C(\epsilon) > 0$ such that
\begin{equation}
\label{equation T energy three charge}
 \int_{\Sigma_t \setminus \tilde{\mathcal{E}}_t^\epsilon}|\partial u|^2 \dVol_{\underline{g}} + \int_{\tilde{\mathcal{E}}_t^\epsilon} |\bar{\partial} u|^2 \dVol_{\underline{g}} \leq C(\epsilon) \int_{\mathcal{E}_t} |N u|^2 \dVol_{\underline{g}} + \int_{\Sigma_0} |\partial u|^2 \dVol_{\underline{g}}
\end{equation}

\end{proposition}
\begin{proof}
 Following identical calculations as in the two charge case, and making use of proposition \ref{proposition T energy current}, we can obtain
\begin{equation}
 \begin{split}
  \imath_{{^{(T)}J}[u]} \dVol_g \big|_{\Sigma_t} &= \frac{1}{2}\bigg( -g(T, n)^{-1}(T u)^2 - g(T, n)\sum_{A = 1}^4 (X_A u)^2 \\
  &\phantom{= \frac{1}{2}\bigg(} - g(T, n)^{-1}\left( (g(T, n))^2 - g(N, N) \right)(X_5 u)^2 \bigg) \dVol_{\underline{g}}
 \end{split}
\end{equation}
The difference with the two charge case is that $g(T, T) = g(N, N) - (g(T,n))^2$ becomes \emph{positive} in the ergoregion $\mathcal{E}$ (indeed, this is the definition of the ergoregion) and so the coefficient of $(X_5 u)^2$ becomes negative in the equation above, within the ergoregion. On the other hand, we have
\begin{equation}
 (g(T, n))^2 - g(N, N) \geq \epsilon \quad \text{in } \tilde{\mathcal{E}}_t^\epsilon
\end{equation}
Finally, we note that, in the ergoregion $\mathcal{E}_t$ we have
\begin{equation*}
 \begin{split}
  (g(T, n))^2 - g(N, N) &< 0 \\
  \Rightarrow g(N,N) > (g(T,n))^2
 \end{split}
\end{equation*}
Since $g(T,n) \sim 1$, $g(N,N)$ is bounded away from zero in the ergoregion, which enables us to replace the $X_5$ derivative by the $N$ derivative in the ergoregion.
\end{proof}

We will also need to perform an energy estimate using the globally null vector field $V$. In order to do this, we need to choose a different basis for the tangent space of $\Sigma_t$, which we will construct with the help of the following definition:
\begin{definition}[The vector field $N_V$]
 Define the vector
\begin{equation}
 N_V := V + g(V, n)n
\end{equation}
\end{definition}

\begin{proposition}[Bounding the norm of $N_V$]
 There is some constant $C > 0$ such that
\begin{equation}
 g(N_V, N_V) \geq C
\end{equation}
\end{proposition}
\begin{proof}
 Since $V$ is null, we have
\begin{equation*}
 0 = -g(V,n)^2 + g(N_V, N_V)
\end{equation*}
 On the other hand, we have
\begin{equation*}
 g(V,n) = g(T, n) + g(Z, n) = g(T, n) \sim 1
\end{equation*}
since $g(Z, n) = 0$.
\end{proof}

\begin{definition}[The vector fields $Y_A$]
 We define an orthonormal basis of vector fields $Y_A$ for $\Sigma_t$, such that $Y_5$ is parallel to $N_V$, i.e.\
\begin{equation}
 Y_5 := g(N_V, N_V)^{-\frac{1}{2}} N_V
\end{equation}
\end{definition}

\begin{definition}[Schematic notation for derivatives adapted to the frame $\{ Y_A\}$]
 We define the following notation:
\begin{equation}
 |\hat{\partial} u|^2 := |V u|^2 + \sum_{A = 1}^4 |Y_A u|^2
\end{equation}
i.e.\ the derivatives $\hat{\partial} u$ do not include the derivative in the $N_V$ direction.
\end{definition}

\begin{proposition}[The $V$-energy estimate on three charge microstate geometries]
 Let $u$ solve the wave equation $\Box_g u = 0$ on a three charge microstate geometry. Then for any $t \geq 0$ we have the following degenerate energy estimate:
\begin{equation}
 \int_{\Sigma_t} |\hat{\partial} u|^2 \dVol_{\underline{g}} \lesssim \int_{\Sigma_0} |\hat{\partial} u|^2 \dVol_{\underline{g}}
\end{equation}
\end{proposition}
\begin{proof}
 Following a very similar set of calculations to those for the $T$-energy estimate (see proposition \ref{proposition T energy current}), we find
 \begin{equation}
  \begin{split}
   \imath_{{^{(V)}J}[u]} \dVol_g \big|_{\Sigma_t} &= \frac{1}{2}\bigg( -g(V, n)^{-1}(V u)^2 - g(V, n)\sum_{A = 1}^4 (Y_A u)^2 \\
   &\phantom{= \frac{1}{2}\bigg(} - g(V, n)^{-1}\left( (g(V, n))^2 - g(N_V, N_V) \right)(Y_5 u)^2 \bigg) \dVol_{\underline{g}}
  \end{split}
 \end{equation}
 However, since $V$ is globally null we find that
 \begin{equation*}
  (g(V,n))^2 = g(N_V, N_V)
 \end{equation*}
 and so the final term in the above equation vanishes, proving the proposition.
\end{proof}

Before we use the Hardy inequality of lemma \ref{lemma Hardy}, we need to show the $R$ can be written in terms of the $\hat{\partial}$ derivatives:
\begin{proposition}
\label{proposition R in terms of dhat}
 The vector field $R$ satisfies $R \in \spn\{ Y_1, Y_2, Y_3, Y_4 \}$, and moreover we can express $R$ as
 \begin{equation}
  R = \sum_{A = 1}^4 \hat{R}^A Y_A
 \end{equation}
 where 
 \begin{equation}
  \sum_{A = 1}^4 |\hat{R}^A| \leq 1
 \end{equation}
 
\end{proposition}
\begin{proof}
 Decomposing in the basis $\{ n, Y_1, \ldots, Y_5\}$ we write
 \begin{equation}
  R = \hat{R}^0 n + \sum_{A = 1}^5 \hat{R}^A Y_A
 \end{equation}
Taking the inner product with $n$, and recalling that $n \propto \upd t^\sharp$ and that $R(t) = 0$, we find that $\hat{R}^0 = 0$. In addition, we have that
\begin{equation*}
 \begin{split}
  \hat{R}^5 &= g(R, Y_5) \\
  &= g(N_V, N_V)^{-\frac{1}{2}} g(R, N_V)
 \end{split}
\end{equation*}
But, since $g(R, n) = 0$ we have that
\begin{equation*}
 g(R, N_V) = g(R, T) = 0
\end{equation*}
where the final equality follows from the explicit form of the metric in equation \eqref{equation three charge metric}.

To finish the proof of the proposition, we simply note that $g(R, R) \sim 1$.
\end{proof}

We also need to show that $N_V$ can be expressed in terms of the vector fields $\Phi$, $\Psi$ and the  derivatives $\hat{\partial}$:
\begin{proposition}
\label{proposition NV}
 The vector field $N_V$ satisfies $N_V \in \spn\{ \Phi, \Psi, Y_1, Y_2, Y_3, Y_4 \}$, and moreover, in the region $\mathcal{E}_t$ we can express $N_V$ as
 \begin{equation}
  N_V = (N_V)^{\Phi} \Phi + (N_V)^{\Psi}\Psi + \sum_{A = 1}^4 (N_V)^A Y_A
 \end{equation}
 where 
\begin{equation}
 |(N_V)^{\Phi}| + |(N_V)^{\Psi}| + \sum_{A = 1}^4 |(N_V)^A| \lesssim 1
\end{equation} 
 in the region $\mathcal{E}_t$.
\end{proposition}
\begin{proof}
 We begin by expressing $\Phi$ in terms of the basis $\{n, Y_1, \ldots, Y_5\}$:
 \begin{equation*}
  \Phi = \Phi^0 n + \sum_{A = 1}^5 \hat{\Phi}^A Y_A
 \end{equation*}
 Taking the inner product with $n$ and recalling that $n \propto \upd t^\sharp$ and $\Phi(t) = 0$, we find that $\Phi^0 = 0$. On the other hand, to find the coefficient $\hat{\Phi}^5$ we compute
 \begin{equation*}
  \begin{split}
   \hat{\Phi}^5 &= g(N_V, N_V)^{-\frac{1}{2}} g(\Phi, N_V) \\
   &= g(N_V, N_V)^{-\frac{1}{2}}g(\Phi, V)
  \end{split}
 \end{equation*}
and, from the explicit form of the metric given in equation \eqref{equation three charge metric} we compute
\begin{equation*}
 g(\Phi, V) =  -\frac{(\tilde{\gamma}_1 + \tilde{\gamma}_2) \eta \sqrt{Q_1 Q_2}}{h f}\sin^2 \theta 
\end{equation*} 
 So we see that, away from $\theta = 0$, we have
 \begin{equation*}
  N_V \in \spn\{ \Phi, Y_1, Y_2, Y_3, Y_4\}
 \end{equation*}
 
 Similarly, we can write
 \begin{equation*}
  \Psi = \Psi^0 n + \sum_{A = 1}^5 \hat{\Psi}^A Y_A
 \end{equation*}
 Once again, since $n \propto \upd t^\sharp$ and $\Psi(t) = 0$ we find that $\Psi^0 = 0$. Now, we have
 \begin{equation*}
  \hat{\Psi}^5 = g(N_V, N_V)^{-\frac{1}{2}}g(\Psi, V)
 \end{equation*}
 and from the explicit form of the metric \eqref{equation three charge metric} we find
 \begin{equation*}
 g(\Psi, V) =  -\frac{(\tilde{\gamma}_1 + \tilde{\gamma}_2) \eta \sqrt{Q_1 Q_2}}{h f}\cos^2 \theta 
\end{equation*} 
 So we see that, away from $\theta = \frac{\pi}{2}$, we have
 \begin{equation*}
  N_V \in \spn\{ \Psi, Y_1, Y_2, Y_3, Y_4\}
 \end{equation*}
 Putting the last few statements together proves the proposition.
\end{proof}

We are now ready to prove uniform boundedness on three charge microstate geometries.

\begin{theorem}[Uniform boundedness of solutions to the wave equation on three charge microstate geometries]
\label{theorem boundedness three charge}
 Let $u$ be a solution to the wave equation $\Box_g u = 0$ on a two charge microstate geometry, with compactly supported initial data on the initial hypersurface $\Sigma_0$. Then, for all $t \geq 0$ we have
\begin{equation}
 \int_{\Sigma_t} |\partial u|^2 \dVol_{\underline{g}} \lesssim \int_{\Sigma_0} \left( |\partial u|^2 + |\partial \Phi u|^2 + |\partial \Psi u|^2 \right)\dVol_{\underline{g}} 
\end{equation}
\end{theorem}
\begin{proof}
 We begin with the $T$-energy estimate for the field $u$ (see proposition \ref{proposition T energy current}). To this, we add the $V$-energy estimate for $u$ multiplied by some large constant $C$. Finally, we add both the $V$-energy estimates for the fields $\Phi u$ and $\Psi u$, also multiplied by the large constant $C$. We note that both $\Phi$ and $\Psi$ are Killing vector fields for the metric \eqref{equation three charge metric}, and so from proposition \ref{proposition commute} the fields $(\Phi u)$ and $(\Psi u)$ also satisfy the wave equation. Thus we find
\begin{equation}
\label{equation boundedness three charge v1}
 \begin{split}
  &\int_{\Sigma_t} \left( |\partial u|^2 + C|\hat{\partial} u|^2 + C|\hat{\partial}\Phi u|^2 + C|\hat{\partial}\Psi u|^2 \right) \dVol_{\underline{g}} \\
  &\lesssim \int_{\mathcal{E}_t} |N_V u|^2 \dVol_{\underline{g}} + \int_{\Sigma_0} \left( |\partial u|^2 + C|\partial \Phi u|^2 + C|\partial \Psi u|^2 \right) \dVol_{\underline{g}}
 \end{split}
\end{equation}
Now, since $u$ is compactly supported on $\Sigma_0$, by the domain of dependence property it is compactly supported on the surface $\Sigma_t$ and so in particular we can apply the Hardy inequality of lemma \ref{lemma Hardy}. We have
\begin{equation*}
 \begin{split}
  \int_{\Sigma_t} \left( |\Phi u|^2 + |\Psi u|^2 \right) \frac{1}{(1+r)^2} \dVol_{\underline{g}} &\lesssim \int_{\Sigma_t} \left( |R \Phi u|^2 + |R \Psi u|^2 \right) \dVol_{\underline{g}} \\
  &\lesssim \int_{\Sigma_t} \left( |\hat{\partial} \Phi u|^2 + |\hat{\partial} \Psi u|^2 \right) \dVol_{\underline{g}}
 \end{split}
\end{equation*}
where the second inequality follows from proposition \ref{proposition R in terms of dhat}.

The region $\mathcal{E}_t$ is compact, so $r$ is bounded above in $\mathcal{E}_t$ and we have
\begin{equation*}
 \int_{\mathcal{E}_t}\left( |\Phi u|^2 + |\Psi u|^2 \right) \dVol_{\underline{g}} \lesssim \int_{\Sigma_t} \left( |\hat{\partial} \Phi u|^2 + |\hat{\partial} \Psi u|^2 \right) \dVol_{\underline{g}}
\end{equation*}
Substituting into equation \eqref{equation boundedness three charge v1} we find
\begin{equation}
\label{equation boundedness three charge v2}
 \begin{split}
  &\int_{\Sigma_t} |\partial u|^2 \dVol_{\underline{g}} + \int_{\mathcal{E}_t} \left( C|\hat{\partial} u|^2 + C|\Phi u|^2 + C|\Psi u|^2 \right) \dVol_{\underline{g}} \\
  &\lesssim \int_{\mathcal{E}_t} |N_V u|^2 \dVol_{\underline{g}} + \int_{\Sigma_0} \left( |\partial u|^2 + C|\partial \Phi u|^2 + C|\partial \Psi u|^2 \right) \dVol_{\underline{g}}
 \end{split}
\end{equation}
Now, we can use proposition \ref{proposition NV} to find
\begin{equation}
 \int_{\mathcal{E}_t} |N_V u|^2 \dVol_{\underline{g}} \lesssim \int_{\mathcal{E}_t} \left( |\hat{\partial} u|^2 + |\Phi u|^2 + |\Psi u|^2 \right) \dVol_{\underline{g}}
\end{equation}
so if we take the constant $C$ to be sufficiently large, then we can absorb the first term on the right hand side of equation \eqref{equation boundedness three charge v2} by the left hand side, proving the theorem.
\end{proof}

We note that in both theorem \ref{theorem boundedness two charge} and \ref{theorem boundedness three charge} the restriction to compactly supported initial data may be replaced by data satisfying, say, $E^{(1)}_{0} < \infty$ by standard density arguments.

\section{Non-decay of waves on three charge microstate geometries}
\label{section non decay}
In the previous section we saw that linear waves on both two and three charge microstate geometries are uniformly bounded, in the sense that their energy is bounded in terms of the initial (higher order) energy. In this section, we will show that the \emph{local} energy of such waves \emph{does not decay} on three charge microstate geometries. This follows from \cite{Friedman1978} (see also the recent work \cite{Moschidis:2016zjy}) but, for completeness, we shall give a slightly more detailed and more explicit construction below.

\begin{theorem}[Non decay of local energy on three charge microstate geometries]
 Let $k \geq 0$ be some constant. Then, on any three charge microstate geometry, there exists a solution $u$ to the wave equation $\Box_g u = 0$, a $T$-invariant open region $\mathcal{U}_t \subset \mathcal{E}_t$ and a positive constant $C_k > 0$ such that, at all times $t \geq 0$, 
 \begin{equation}
  \int_{\mathcal{U}_t} |N u|^2 \dVol_{\underline{g}} \geq C_k E^{(k)}_0(u)
 \end{equation}
 Moreover, the initial data for $u$ can be chosen to be smooth and compactly supported.
 
\end{theorem}
\begin{proof}
 Applying the $T$-energy estimate (using proposition \ref{proposition T energy current}) we find that
\begin{equation}
 \begin{split}
  &\int_{\Sigma_t} -\frac{1}{2}g(T,n)^{-1}\bigg( (T u)^2 + g(T, n)^2\sum_{A = 1}^4 (X_A u)^2 - g(T, T)(X_5 u)^2 \bigg) \dVol_{\underline{g}} \\
  &= \int_{\Sigma_0} -\frac{1}{2}g(T,n)^{-1}\bigg( (T u)^2 + g(T, n)^2\sum_{A = 1}^4 (X_A u)^2 - g(T, T)(X_5 u)^2 \bigg) \dVol_{\underline{g}}
 \end{split}
\end{equation}
In particular, by separating the surface $\Sigma_t$ into the ergoregion $\mathcal{E}_t$ and its complement, we find
\begin{equation}
\label{equation non decay v1}
 \begin{split}
  &\int_{\mathcal{E}_t} -g(T, n)^{-1} g(T,T)(X_5 u)^2 \dVol_{\underline{g}} \\
  &= \int_{\Sigma_t} \left( -g(T,n)^{-1} (T u)^2  - g(T,n)\sum_{A = 1}^4 (X_A u)^2 \right) \dVol_{\underline{g}} 
  + \int_{\Sigma_t \setminus \mathcal{E}} -g(T, n)^{-1} g(T,T)(X_5 u)^2 \dVol_{\underline{g}} \\
  & + \int_{\Sigma_0} g(T,n)^{-1}\bigg( (T u)^2 + g(T, n)^2\sum_{A = 1}^4 (X_A u)^2 - g(T, T)(X_5 u)^2 \bigg) \dVol_{\underline{g}}
 \end{split}
\end{equation}
Now, since the hypersurface $\Sigma_t$ is uniformly spacelike, and 
\begin{equation*}
 g(T, n) = -\left(- g^{-1}(\upd t, \upd t) \right)^{-\frac{1}{2}}
\end{equation*}
we have that $g(T, n) < 0$, and so the first integral on the right hand side of \eqref{equation non decay v1} is non-negative. Additionally, outside the ergoregion we have $g(T,T) \leq 0$ and so the second integral on the right hand side is non-negative. We conclude that
\begin{equation}
\label{equation non decay v2}
 \begin{split}
  &\int_{\mathcal{E}_t} g(T,T)(X_5 u)^2 \dVol_{\underline{g}} \gtrsim \int_{\Sigma_0} \bigg( -(T u)^2 - g(T, n)^2\sum_{A = 1}^4 (X_A u)^2 + g(T, T)(X_5 u)^2 \bigg) \dVol_{\underline{g}}
 \end{split}
\end{equation}
and so, if we can find initial data such that the right hand side of \eqref{equation non decay v2} is strictly positive, then the local energy of $u$ in the ergoregion at any time $t \geq 0$ will be bounded away from zero. This corresponds to constructing initial data such that the energy associated with the $T$ vector field is strictly negative.

On $\Sigma_0$, we can freely prescribe both $u$ and $(n u)$, which we do in the following way. Define the submanifold $(\mathcal{U}_1)_t^\epsilon \subset \mathcal{E}_t$ by
\begin{equation*}
 (\mathcal{U}_1)_t^\epsilon := \{ x \in \mathcal{E}_t \ \big| \ z = 0, g(T,T) > \epsilon \}
\end{equation*}
Note that the vector field $N$ is transverse to the submanifold $(\mathcal{U}_1)_t^\epsilon$, indeed, we have
\begin{equation}
 N = (-g^{tt})\left(g^{tz}Z + g^{t\phi}\Phi + g^{t\psi}\Psi \right)
\end{equation}
and $g^{tt} = g^{-1}(\upd t, \upd t) \sim 1$ while
\begin{equation}
 g^{tz} = -\frac{Q_p}{hf} \left( 1 + \frac{Q_1 + Q_2}{r^2 + (\tilde{\gamma}_1 + \tilde{\gamma}_2)^2 \eta)} \right)
\end{equation}
which is also bounded away from zero in the ergoregion $\mathcal{E}_t$. Hence, we can define the region $(\mathcal{U}_2)_t^{(\epsilon, \delta)}$ as the set of points in $\mathcal{M}$ which can be reached from $(\mathcal{U}_1)_t^\epsilon$ by moving along the integral curves of $N$ (in either direction) a distance of at most $\delta$. We can define the parameter $s$ by
\begin{equation}
 \begin{split}
  X_5(s) = 1 \\
  s\big|_{(\mathcal{U}_1)_t^\epsilon} = 0
 \end{split}
\end{equation}
Then, for all $\epsilon > 0$, we see that for all sufficiently small $\delta$, the region $(\mathcal{U}_2)_t^{(\epsilon, \delta)}$ is an open region (as a subset of $\Sigma_t$), strictly contained within the ergoregion $\mathcal{E}_t^{(\epsilon, \delta)}$.

Now, we define a smooth cut-off function $\chi$ such that
\begin{equation}
 \begin{split}
  \chi\big|_{\partial (\mathcal{U}_2)_t^{(\epsilon, \delta)}} &= 0 \\
  \chi &= 0 \quad \text{on } \Sigma_t \setminus (\mathcal{U}_2)_t^{(\epsilon, \delta)} \\
  \chi &= 1 \quad \text{on } (\mathcal{U}_2)_t^{(2\epsilon, \delta/2)}
 \end{split} 
\end{equation}
Note that, for $\epsilon$ sufficiently small, the set $(\mathcal{U}_2)_t^{(2\epsilon, \delta/2)}$ is nonempty and is strictly contained within the set $(\mathcal{U}_2)_t^{(\epsilon, \delta)}$.

Finally, we are ready to define the initial data for the wave equation. We define the initial data
\begin{equation}
\label{equation initial data}
 \begin{split}
  u\big|_{\Sigma_t} &= \chi \cdot \sin(M s) \\
  (nu)\big|_{\Sigma_t} &= (g(T, n))^{-1} N\left( \chi \cdot \sin(M s) \right) \\
  &= (g(T, n))^{-1} \left( (N \chi) \sin(M s) + M\chi \cos(M s) \right)
 \end{split}
\end{equation}
for some large constant $M$ to be determined later. Then we find 
\begin{equation*}
 (Tu)\big|_{\Sigma_t} = 0
\end{equation*}
and the right hand side of equation \eqref{equation non decay v2} is given by
\begin{equation}
 \begin{split}
  &\int_{\Sigma_0} \bigg( -(T u)^2 - g(T, n)^2\sum_{A = 1}^4 (X_A u)^2 + g(T, T)(X_5 u)^2 \bigg) \dVol_{\underline{g}} \\
  &= \int_{\Sigma_0} \bigg( -g(T,n)^2 \sin^2(Ms) \sum_{A = 1}^4 (X_a \chi)^2 \\
  &\phantom{=} + g(T,T)\left( M^2 \chi \sin^2 (Ms) + 2 M \chi (X_5 \chi) \sin(Ms) \cos(Ms) + (X_5(X_5 \chi))\sin^2(Ms) \right) \bigg)\dVol_{\underline{g}} \\
  &\gtrsim \int_{\Sigma_0} \bigg( M^2 \chi \sin^2(M s) - C\left(|\partial \partial \chi| + |\partial \chi|(1+ M) \right) \bigg) \dVol_{\underline{g}}
 \end{split}
\end{equation}
for some large constant $C$. Since $\chi$ is smooth, if we take $M$ sufficiently large then this is positive. Moreover, note that the initial data defined in \eqref{equation initial data} is smooth and compactly supported, proving the theorem.

\end{proof}

\section{Slow decay of waves on two charge microstate geometries}
\label{section slow decay}

Unlike in the three charge case, two charge microstate geometries do not posses an ergoregion, but only a submanifold $\mathcal{S}$ such that $T$ is null on $\mathcal{S}$. As such, in contrast to the three charge case, we cannot use the construction of Friedman to produce waves whose local energy does not decay. Instead, in this section we will construct \emph{quasimodes}: smooth, compactly supported \emph{approximate} solutions to the wave equation which do not decay. Since these approximate solutions solve the wave equation with a very small error, they can be used to contradict any uniform decay statement with a sufficiently fast decay rate.

Our approach in this section closely follows the approach first used in \cite{Holzegel2014}: we use the separability of the wave equation to construct mode solutions in some bounded region, where we artificially impose (Dirichlet) boundary conditions on the boundaries (which will be surfaces of constant $r$). We then continue these functions in some smooth way, so that they vanish in a slightly larger region. The idea is that, with a judicious placement of the boundaries, these functions will be very close to solutions of the wave equation.

There are, however, several major differences between the present work and that of \cite{Holzegel2014}. First, we are aiming to show \emph{slower than logarithmic decay}, which means that the discrepancy between the quasimodes we construct and actual solutions to the wave equation must be \emph{super exponentially} suppressed. This is possible because of the structure of the potential, which has a local minimum at (in the high angular frequency limit) exactly zero. To exploit this, we must construct mode solutions with time frequencies which are \emph{uniformly bounded} in terms of the angular frequency. This contrasts with the approach of \cite{Holzegel2014}, in which the time frequency grows in proportion to the angular frequency. We must also then adjust the position of the boundaries where we impose Dirichlet boundary conditions: we choose these to be at some value of $r$ which grows with the frequency, rather than at some fixed value of $r$ as in \cite{Holzegel2014}. The combination of these two features will allow us to obtain the desired super-exponentially small errors.

We also encounter an additional technical difficulty: at each fixed frequency, the effective potential (after separating variables) diverges at $r = 0$. This must be carefully handled in order to ensure that we retain the desired control over the eigenvalues, which correspond to the time frequencies.

Finally, we note that, as in \cite{Holzegel2014}, to construct the quasimodes we need to solve a nonlinear eigenvalue problem. As in \cite{Holzegel2014}, we can do this by first solving a related \emph{linear} eigenvalue problem, and then using perturbative arguments. In \cite{Holzegel2014} certain monotonicity properties of the potential were used, along with a continuity argument, to handle this perturbative step. In the present case, the relative smallness of the eigenvalues we are considering actually makes things easier: we find that we are able to appeal directly to the implicit function theorem to obtain a solution to the nonlinear eigenvalue problem, given a solution to the corresponding linear problem, at least for sufficiently large angular momentum.

\subsection{Construction of modes with bounded frequency}

This section is devoted to the proof of the following theorem:

\begin{theorem}[Existence of quasimodes with bounded frequencies]
\label{theorem existence of modes}
 On any two charge microstate geometry, there exist ``mode solutions'': regular solutions to the wave equation $\Box_g u = 0$, where Dirichlet boundary conditions are imposed at $r = R R_z \mu_\ell$ (in particular, these solutions are regular at $r = 0$), where $R$ is some suitably large constant to be fixed later. These mode solutions are of the form
 \begin{equation}
  u = e^{-i\omega t + i\lambda z + i m_\psi \psi + i m_\phi \phi} u_r(r) u_\theta(\theta)
 \end{equation}
 for some $\omega \in \mathbb{R}$ and integers $(\lambda / R_z)$, $m_\psi$ and $m_\phi$. $\mu_\ell$ is defined to be the $\ell$-th eigenvalue associated with the eigenvalue equation for $u_\theta$ (see equation \eqref{equation defining mu ell} below), and we have
 \begin{equation}
  \ell = \mu_\ell - 1 + \mathcal{O}\left((\mu_\ell)^{-1}\right)
 \end{equation}
 Finally, and crucially, these mode solutions can be chosen such that the frequency $\omega$ \emph{does not scale with} $\ell$. That is, if we choose some constant $\omega_{\text{max}}$ sufficiently large, then for all sufficiently large $\ell$ there exists a mode solution with
 \begin{equation}
  \omega \leq \omega_{\text{max}}
 \end{equation}
\end{theorem}

The proof of theorem \ref{theorem existence of modes} is rather technical, so we shall first outline the structure of the proof. We begin by establishing the existence of mode solutions with the required properties for a related \emph{linear} eigenvalue problem. We then study a family of eigenvalue problems, labeled by a parameter $b$, which continuously transition between the linear problem and the true eigenvalue problem we wish to study. In particular, $b = 0$ labels the linear eigenvalue problem, and $b = 1$ labels the related nonlinear problem. While the desired properties of the mode solutions hold, we are able to use the implicit function theorem to show that mode solutions exist at slightly larger values of $b$ than those we have already obtained. We are then able to show that the desired properties of the mode solutions also extend to larger values of $b$, as long as the mode solutions exist. We are thus able to bring $b$ all the way up to $b=1$, establishing the desired result.

\subsubsection{Deriving the eigenvalue problem}

As mentioned above, we shall search for solutions of the wave equation on two charge microstate geometries which are of the form
\begin{equation}
 u = e^{-i\omega t + i\lambda z + i m_\psi \psi + i m_\phi \phi} u_r(r) u_\theta(\theta)
\end{equation}
We are trying to find solutions which are localised near the surface on which null geodesics are stably trapped (i.e.\ the submanifold $\mathcal{S}$), which is defined by $r = 0$ and $\theta = \frac{\pi}{2}$. The corresponding null geodesics have vanishing momentum in the $\psi$ direction, and so we set
\begin{equation}
 m_\psi = 0
\end{equation}

Now, following \cite{Eperon2016}, we shall work with dimensionless variables
\begin{equation}
\label{equation dimensionless variables}
 \begin{split}
  y &:= \frac{r}{R_z} \\
  \tilde{\omega} &:= \omega R_z \\
  \tilde{\lambda} &:= \lambda R_z \\
  s &:= \frac{\sqrt{Q_1 Q_2}}{(R_z)^2} \\
 \end{split}
\end{equation}

We also define $(\mu_\ell)^2$ to be the $\ell$-th eigenvalue associated with the $u_\theta$ equation. To be specific, we let $u_{\theta, \ell}$ be the eigenfunction associated with the $\ell$-th eigenvalue of the $\theta$ equation, i.e.\
\begin{equation}
\label{equation defining mu ell}
 \begin{split}
  -\frac{1}{\sin 2\theta} \frac{\upd}{\upd \theta} \left( \sin 2\theta \frac{\upd u_{\theta, \ell}}{\upd \theta} \right) + \left( \frac{(m_\phi)^2}{\sin^2 \theta} + (\tilde{\lambda}^2 - \tilde{\omega}^2)\frac{a^2}{(R_z)^2} \cos^2 \theta \right)u_{\theta, \ell} = (\mu_\ell)^2 u_{\theta, \ell}
 \end{split}
\end{equation}

Additionally, and also motivated by the geometric optics approximation, we shall set
\begin{equation}
 m_\phi = -\ell
\end{equation}

In terms of these variables, the equation satisfied by the radial wavefunction $u_r(r)$ can be written in self adjoint form as
\begin{equation}
  -y(y^2 + s^2) \frac{\upd}{\upd y}\left( y(y^2 + s^2) \frac{d u_r}{dy} \right) + V_1(y; \ell, \omega)u_r = 0
\end{equation}
where
\begin{equation}
 \begin{split}
  V_1(y; \mu_\ell, \tilde{\omega}, \tilde{\lambda}) := - (\tilde{\omega}^2 - \tilde{\lambda}^2) y^6 + A y^4 + B y^2 + D
 \end{split}
\end{equation}
and where the coefficients are as follows:
\begin{equation}
 \begin{split}
  A &:= (\mu_\ell)^2 - (\tilde{\omega}^2 - \tilde{\lambda}^2)\left( s^2 + \frac{Q_1 + Q_2}{(R_z)^2} \right) \\
  B &:= s^2\left( (\mu_\ell)^2 - \ell^2 \right) - 2s^2 \tilde{\omega} \ell - s^2(\tilde{\omega}^2 - \tilde{\lambda}^2)\left(1 + \frac{Q_1 + Q_2}{(R_z)^2} \right) \\
  D &:= s^4 \tilde{\lambda}^2
 \end{split}
\end{equation}

We define the variable $w$ by
\begin{equation}
  w := \left(y(y^2 + s^2)\right)^{\frac{1}{2}} u_r
\end{equation}
then we find that $w$ satisfies
\begin{equation}
\label{equation V_2}
  - \frac{\upd^2}{\upd y^2}w + V_2(r; \mu_\ell, \tilde{\omega}, \tilde{\lambda})w = 0
\end{equation}
with $V_2$ related to $V_1$ by
\begin{equation}
  V_2(y; \mu_\ell, \tilde{\omega}, \tilde{\lambda}) := \frac{V_1(y; \mu_\ell, \tilde{\omega}, \tilde{\lambda})}{y^2(y^2 + s^2)^2} + \frac{ 3y^4 + 6s^2 y^2 - s^4}{4 y^2 (y^2 + s^2)^2}
\end{equation}
Note that, as $y \rightarrow 0$, we have
\begin{equation*}
 V_2 \sim \frac{ 4 \tilde{\lambda}^2 - 1}{4 y^2}
\end{equation*}
Now, we can re-write equation \eqref{equation V_2} in the form
\begin{equation}
\label{equation radial eigenvalue V3}
 -\frac{\upd^2}{\upd y^2}w + V_3 w = 2 s^{-2}\ell \tilde{\omega} w
\end{equation}
where
\begin{equation}
 V_3(y; \mu_\ell, \tilde{\omega}, \tilde{\lambda}) := V_2 + 2s^{-2}\ell \tilde{\omega} 
\end{equation}
Note that, as $y \rightarrow \infty$ we have
\begin{equation}
 V_3(y; \mu_\ell, \tilde{\omega}, \tilde{\lambda}) \rightarrow - (\tilde{\omega}^2 - \tilde{\lambda}^2) + 2s^{-2} \tilde{\omega} \ell
\end{equation}

The above is summarised in the following:
\begin{proposition}
	Let $u$ be of the form
	\begin{equation}
		u
		=
		\exp\left( 
			-i\frac{\tilde{\omega}}{R_z} t
			+ i \frac{\tilde{\lambda}}{R_z} z
			- i \ell \phi
		\right) \left( y(y^2 + s^2) \right)^{-\frac{1}{2}} w(y) u_{\theta,\ell}(\theta)
	\end{equation}
	where the dimensionless variables are defined in equation \eqref{equation dimensionless variables}. Suppose that $u_{\theta,\ell}$ solves the eigenvalue problem \eqref{equation defining mu ell}, and suppose that $w$ satisfies equation \eqref{equation radial eigenvalue V3}.
	
	Then $u$ solves the wave equation $\Box_g u = 0$.
\end{proposition}

\subsubsection{The operator \texorpdfstring{$\mathcal{H}$}{H} and its properties}

Let us define the Hermitian operator $\mathcal{H}$ by
\begin{equation}
 \mathcal{H}(w) := -\frac{\upd^2}{\upd y^2}w + V_3(y; \mu_\ell, \tilde{\omega}, \tilde{\lambda}) w
\end{equation}
for $w$ in some suitable function space. Our aim is to show that $\mathcal{H}$ admits eigenvalues which don't scale with $\ell$, but as a preliminary step we need to establish that $\mathcal{H}$, defined to act on a suitable function space, has compact resolvent and trivial kernel. Since the effective potential $V_3$ diverges as $y \rightarrow 0$, this is not obvious.

\begin{proposition}
\label{proposition H maps function spaces and has compact resolvent}

The for any $R > 0$, the operator $\mathcal{H}$ is a linear map
\begin{equation*}
	\mathcal{H} \ : \ H^1_0([0, R \mu_\ell]) \rightarrow H^{-1}_0([0, R\mu_\ell])
\end{equation*}
with compact resolvent.

\end{proposition}

\begin{proof}
	Recall that, since functions in $H^1_0[0, R \mu_\ell]$ are limits of sequences of compactly supported functions in $(0, R\mu_\ell)$, the operator $\mathcal{H}$ is Hermitian. Moreover, the image of $H^1_0([0, R \mu_\ell])$ does in fact lie in the space $H^{-1}_0([0, R\mu_\ell])$, as we shall show below. Let $w_1 \in H_0^1([0, R\mu_\ell])$, then we have
	\begin{equation}
		|| \mathcal{H}w_1 ||^2_{H_0^{-1}} = \sup_{w_2 \in H^1_0} \frac{ \langle w_2, \mathcal{H} w_1 \rangle^2}{||w_2||^2_{H^1}}
	\end{equation}
	where $\langle \cdot , \cdot \rangle$ is the standard $L^2$ inner product on the interval $[0, R\mu_\ell]$, and we have omitted the interval on which the norms are defined for visual clarity. But we have
	\begin{equation}
	\begin{split}
		\langle w_2, \mathcal{H} w_1 \rangle^2 
		&= 
		\left( \int_0^{R \mu_\ell} \left( \left( \frac{\partial w_1}{\partial y} \right) \left( \frac{\partial w_2}{\partial y}\right) + V_3 w_1 w_2 \right) \upd y \right)^2 		\\
		&\lesssim 
		||w_1||^2_{H^1} ||w_2||^2_{H^1} + \left(\int_0^{R \mu_\ell} \frac{1}{y^2} |w_1| |w_2| \upd y \right)^2
	\end{split}	
	\end{equation}
	where we have made use of the Cauchy-Schwartz inequality, and the implicit constant in the above inequality is allowed to depend on $\mu_\ell$, $\tilde{\lambda}$ in addition to the metric parameters. Now, using Cauchy-Schwartz again, we have
	\begin{equation*}	
	\begin{split}
		\left(\int_0^{R \mu_\ell} \frac{1}{y^2} w_1 w_2 \upd y \right)^2 
		\lesssim 
		\left(\int_0^{R \mu_\ell} \frac{1}{y^2} (w_1)^2\upd y \right)\left(\int_0^{R \mu_\ell} \frac{1}{y^2} (w_2)^2 \upd y \right)
	\end{split}
	\end{equation*}
	and, for $w \in H^1_0([0, R \mu_\ell])$ we can integrate by parts to write
	\begin{equation*}
	\begin{split}
		\int_0^{R\mu_\ell} \frac{1}{y^2} w^2 \upd y &= \int_0^{R\mu_\ell} -\partial_y \left( \frac{1}{y} \right) w^2 \upd y 
		\\
		&=
		\int_0^{R\mu_\ell} \frac{2}{y} w (\partial_y w) \upd y 
		\\
		&\leq \int_0^{R \mu_\ell} \left( \frac{1}{2y^2} w^2 + 2 (\partial_y w)^2 \right)\upd y
	\end{split}
	\end{equation*}
	and so, for $w \in H^1_0([0, R\mu_\ell])$ we have
	\begin{equation}
		\int_0^{R\mu_\ell} \frac{1}{y^2} w^2 \upd y \leq 4\int_0^{R\mu_\ell} (\partial_y w)^2 \upd y \lesssim ||w||_{H^1([0, R\mu_\ell])}^2
	\end{equation}
	all of which means that
	\begin{equation*}
		|| \mathcal{H}w_1 ||^2_{H_0^{-1}} \lesssim ||w_1||^2_{H^1}
	\end{equation*}
	meaning that the image of $H^1_0([0, R\mu_\ell])$ by the operator $\mathcal{H}$ does lie in the space $H^{-1}_0([0, R\mu_\ell])$ as promised. In fact, the same calculation also shows that the operator $\mathcal{H}$ has compact resolvent.
\end{proof}

We also want to show that the operator $\mathcal{H}$ is positive, in the sense that, for any $w \in H^1_0([0, R\mu_\ell])$ we have
\begin{equation*}
 \langle w, \mathcal{H} w \rangle > 0
\end{equation*}
In fact, we can only do this for large values of $\mu_\ell$, assuming also a lower bound on $\mu_\ell$ in terms of $\tilde{\omega}$ and $\tilde{\lambda}$ (recall that $\mu_\ell$ is an eigenvalue of the problem \eqref{equation defining mu ell}, in which $\tilde{\omega}$ and $\tilde{\lambda}$ appear):

\begin{proposition}
	Suppose that there exist values of $\mu_\ell$, eigenvalues of the problem \eqref{equation defining mu ell}, satisfying the bound:
	\begin{equation}
	\label{equation bootstrap 1}
		\left( 
			(\mu_\ell)^2 - \ell^2 
		\right) 
		- (\tilde{\omega}^2 - \tilde{\lambda}^2)\left(1 + \frac{Q_1 + Q_2}{(R_z)^2} \right) 
		> c > 0
	\end{equation}
	where $c$ is some fixed positive constant (independent of all of the other variables).
	
	Then, for such values of $\mu_\ell$, we have $\langle w, \mathcal{H} w \rangle > c|| w ||_{L^2([0, R\mu_\ell])}$ (recall that the operator $\mathcal{H}$ depends on $\mu_\ell$ through the form of the potential $V_3$). In particular, $\langle w, \mathcal{H} w \rangle > 0$ for all nonzero $w \in H_0^1([0, R\mu_\ell])$, i.e.\ $\mathcal{H}$ is a positive operator.
\end{proposition}

Note that we shall need to justify the bound \eqref{equation bootstrap 1} later, i.e.\ we will need to construct these eigenvalues of the problem \eqref{equation defining mu ell}.

\begin{proof}

	To prove that $\mathcal{H}$ is positive, note that for sufficiently large $\mu_\ell$ (assuming the bound \eqref{equation bootstrap 1} we have
	\begin{equation*}
		\langle w , \mathcal{H} w \rangle 
		> 
		\int_0^{R\mu_\ell} \left( \left(\frac{\upd w}{\upd y} \right)^2 - \frac{1}{4 y^2} w^2 + c w^2 \right) \upd y
	\end{equation*}
	but we have already seen that, for $w \in H^1_0[0, R\mu_\ell]$, 
	\begin{equation*}
		\int_0^{R\mu_\ell} \frac{1}{y^2} w^2 \upd y \leq 4\int_0^{R\mu_\ell} (\partial_y w)^2 \upd y
	\end{equation*}
	so $\mathcal{H}$ is positive. In particular, this implies the lower bound: for all sufficiently large $\mu_\ell$ obeying the bound \eqref{equation bootstrap 1},  $\tilde{\omega} \geq \tilde{\omega}_{\text{min}}$ for some $\tilde{\omega}_{\text{min}} > 0$. Note that, if $\tilde{\omega}_{\text{min}}$ is chosen to be the smallest eigenvalue of $\mathcal{H}$, then, as we will see later, $\tilde{\omega}_{\text{min}}$ can be bounded above \emph{uniformly} in $\ell$. We will not need need a lower bound for $\tilde{\omega}_{\text{min}}$.
	
\end{proof}

Summarizing the above calculations we see that, for values of $\mu_\ell$ obeying the bound \eqref{equation bootstrap 1}, $\mathcal{H}$ is a positive Hermitian operator with compact resolvent,  from $H^1_0([0, R\mu_\ell])$ to $H^{-1}_0([0, R\mu_\ell])$. The space $H^1_0([0, R\mu_\ell])$ therefore admits a basis of eigenfunctions of $\mathcal{H}$, whose associated eigenvalues are positive and can be listed in ascending order.

\subsubsection{The linear eigenvalue problem}

Now that we have established these basic properties of the operator $\mathcal{H}$, we wish to bound the associated eigenvalues. Specifically, we would like to solve the following problem:
\begin{equation}
\label{equation nonlinear eigenvalue}
 \begin{split}
  -\frac{\upd^2}{\upd y^2} w + V_3(y; \mu_\ell, \tilde{\omega}, \tilde{\lambda}) w &= 2s^{-2}\ell \tilde{\omega}w \\
  w\big|_{y = 0} = w\big|_{y = R\mu_\ell} &= 0
 \end{split}
\end{equation}
Note that, at fixed values of $\mu_\ell$ and $\tilde{\lambda}$ this can be regarded as a \emph{nonlinear} eigenvalue problem: the unknown $\tilde{\omega}$ plays the role of an eigenvalue (in fact, the eigenvalue is a constant multiple of $\tilde{\omega}$), but the effective potential $V_3$ \emph{also} depends on the value $\tilde{\omega}$.

Before tackling this, we shall first consider the related \emph{linear} eigenvalue problem, obtained by replacing $\tilde{\omega}$ when it appears in the effective potential $V_3$ by some fixed constant $\tilde{\omega}_0$, i.e.\
\begin{equation}
\label{equation linear eigenvalue}
 \begin{split}
  -\frac{\upd^2}{\upd y^2} w + V_3(y; \mu_\ell, \tilde{\omega}_0, \tilde{\lambda}) w &= 2s^{-2}\ell \tilde{\omega}w \\
  w\big|_{y = 0} = w\big|_{y = R\mu_\ell} &= 0
 \end{split}
\end{equation}

We will also forget the dependence of $\mu_\ell$ on $\omega$ (through equation \eqref{equation defining mu ell}) and, for the purpose of solving this linear equation, treat $\mu_\ell$ as an independent constant.

Regarding this linear eigenvalue problem, we will establish the following result:
\begin{proposition}
\label{proposition linear eigenvalue problem}
	Suppose that, for sufficiently large $\ell$, we have
	\begin{equation*}
	 \label{equation bootstrap 2}
		\ell = \mu_\ell - 1 + \mathcal{O}\left( (\mu_\ell)^{-1} \right)
	\end{equation*}
	(note that, if this is true, then the bound \eqref{equation bootstrap 1} with $\tilde{\omega}$ replaced by $\tilde{\omega}_0$ will certainly be satisfied for sufficiently large $\ell$).
	
	Then there exists $\tilde{\omega}_{(\text{max})} = \tilde{\omega}_{(\text{max})}(s, \tilde{\lambda}) > 0$, which is \emph{independent of $\mu_\ell$}, such that for all sufficiently large $\ell$ the linear eigenvalue problem \eqref{equation linear eigenvalue} has at least one solution with an associated eigenvalue $2s^{-2} \ell \tilde{\omega}$, where
	\begin{equation*}
		4 \leq \tilde{\omega} \leq \tilde{\omega}_{\text{max}}
	\end{equation*}
	
\end{proposition}

\begin{proof}

In a similar way to \cite{Holzegel2014} and \cite{Keir2016}, we will place a lower bound on the number of eigenvalues below some threshold value\footnote{However, unlike \cite{Holzegel2014} and \cite{Keir2016}, we do not also prove an upper bound on the number of eigenvalues. Note that this upper bound was not necessary for the construction of quasimodes or the bounds on the uniform decay rates.}. Note that we will look for eigenvalues which scale like $\ell$ rather than $\ell^2$, which was the rate used in the two papers cited above. Assuming that $\mu_\ell \gg 1$, we split the region $y \leq R\mu_\ell $ into three sub-regions: 
\begin{equation}
 \begin{split}
  R_1 &:= \left\{ 0 \leq y < \frac{1}{2}(\mu_\ell)^{-\frac{1}{2}} \right\} \\
  R_2 &:= \left\{ \frac{1}{2}(\mu_\ell)^{-\frac{1}{2}} \leq y < (\mu_\ell)^{-\frac{1}{2}} \right\} \\
  R_3 &:= \left\{ (\mu_\ell)^{-\frac{1}{2}} \leq y \leq R\mu_\ell \right\}
 \end{split}
\end{equation}

Let $N(\tilde{\omega}_{\text{max}})$ be the number of eigenvalues $\tilde{\omega}$ below the threshold $\left(2s^{-2}\mu_\ell\tilde{\omega}_{\text{max}}\right)$ for the linear eigenvalue problem \eqref{equation linear eigenvalue}. Let $N_{R_2}(\tilde{\omega}_{\text{max}})$ be the number of eigenvalues below the same threshold for the related Dirichlet problem on the region $R_2$, obtained by replacing the effective potential $V_3$ by its maximum value in the region $R_2$, i.e.\ we search for solutions to the problem
\begin{equation}
\label{equation linear eigenvalue reduced}
  \begin{split}
   -\frac{\upd^2}{\upd y^2} w + V_{\text{max}}(\mu_\ell, \tilde{\omega}_0, \tilde{\lambda}) w &= 2s^{-2}\ell \tilde{\omega}w \\
   w\big|_{\partial R_2} & = 0
  \end{split}
\end{equation}
where $V_{\text{max}}$ is the maximum value of $V_3$ in the region $R_2$:
\begin{equation}
 V_{\text{max}}( \mu_\ell, \tilde{\omega}_0, \tilde{\lambda}) := \sup_{r\in R_2} V_3(y; \mu_\ell, \tilde{\omega}_0, \tilde{\lambda})
\end{equation}
and where we are looking for eigenvalues satisfying the bound
\begin{equation}
 \tilde{\omega} \leq \tilde{\omega}_{\text{max}}
\end{equation}

We claim that $N_{R_2}(\tilde{\omega}_{\text{max}}) \leq N(\tilde{\omega}_{\text{max}})$. To see this, consider the variational characterisation of the eigenvalues. For the original problem \eqref{equation linear eigenvalue}, using the minimax principle, the $n$-th eigenvalue can be characterised as
\begin{equation}
 2s^{-2}\ell\omega_n = \inf_{ \substack{
    \{f_1, f_2, \ldots, f_n\} \\
    f_i \in H^1_{0}([0, R \mu_\ell]) \quad \forall 0 \leq i \leq n \\
    ||f_i||_{L^2([0, R \mu_\ell])} \neq 0 \quad \forall 0 \leq i \leq n \\
    \langle f_i \, , \, f_j \rangle = 0 \quad \forall i \neq j }}
    \max_{i \leq n}
    \frac{ \int_0^{R \mu_\ell} (f_i)(\mathcal{H} f_i) \upd y }{ ||f_i||^2_{L^2([0, R\mu_\ell])} }
\end{equation}
where $\langle \cdot , \cdot \rangle$ represents the standard $L^2$ inner product on the interval $[0, R\mu_\ell]$. In other words, to find the $n$-th eigenvalue we need to consider sets of $n$ mutually orthogonal (in $L^2$) functions, find the largest Rayleigh quotient among them, and then minimize this among all such sets of functions.

Similarly, the $n$-th eigenvalue for the problem \eqref{equation linear eigenvalue reduced} on $R_2$, which we label as $\tilde{\omega}^+_n$, can be characterised (using the min-max theorem and the fact that the resolvent of $\mathcal{H}$ is compact) as
\begin{equation}
 2s^{-2}\ell\tilde{\omega}^+_n = \inf_{ \substack{
    \{f_1, f_2, \ldots, f_n\} \\
    f_i \in H^1_0(R_2) \quad \forall 0 \leq i \leq n \\
    ||f_i||_{L^2(R_2)} \neq 0 \quad \forall 0 \leq i \leq n \\
    \langle f_i \, , \, f_j \rangle = 0 \quad \forall i \neq j }}
    \max_{i \leq n}
    \frac{ \int_0^{R \mu_\ell} (f_i)(\mathcal{H} f_i) \upd y }{ ||f_i||^2_{L^2(R_2)} }
\end{equation}

It is clear that $\tilde{\omega}^+_n \geq \tilde{\omega}_n$, both because $V_{\text{max}} \geq V_3$ in the region in question, and because to find $\tilde{\omega}_n$ the $f_i$ are allowed to range over a strictly larger function space, that is, sets of mutually orthogonal functions, all lying in $H^1_0([0, R\mu_\ell])$ rather than sets of mutually orthogonal functions that lie in $H^1_0(R_2)$. Thus $N_{R_2}(\tilde{\omega}_{\text{max}}) \leq N(\tilde{\omega}_{\text{max}})$. In particular, since we are only interested in a lower bound on the number of eigenvalues, this allows us to safely neglect regions 1 and 3 and focus only on the central region.

Now, in the region $R_2$, for fixed values of $\tilde{\omega}_0$ and $\tilde{\lambda}$, using the explicit form of the potential $V_3$ together with the bounds $\frac{1}{2} \mu_\ell^{-\frac{1}{2}} \leq y \leq \mu_\ell^{-\frac{1}{2}}$, that
\begin{equation}
\label{equation potential in region y sim muell -1/2}
 \begin{split}
  V_3(\mu_\ell, \tilde{\omega}_0, \tilde{\lambda}) &= s^{-2}\left( (\mu_\ell)^2 - \ell^2 \right) + \left( s^{-4} + 4\tilde{\lambda}^2 \right) \mu_\ell + \mathcal{O}(1) \\
  &= \left( s^{-4} + 2s^{-2} + 4\tilde{\lambda}^2 \right) \mu_\ell + \mathcal{O}(1)
 \end{split}
\end{equation}
and so, for sufficiently large $\mu_\ell$ we have the bound
\begin{equation}
\label{equation bound V3 in R2}
 V_3(\mu_\ell, \tilde{\omega}_0, \tilde{\lambda}) \leq 2\left( s^{-4} + 2s^{-2} + 4\tilde{\lambda}^2 \right) \ell
\end{equation}

Now, we can explicitly solve the problem \eqref{equation linear eigenvalue reduced} since the effective potential is just a fixed constant. Using the bound \eqref{equation bound V3 in R2} for $V_3$ in the region $R_2$ we find that the $n$-th eigenvalue satisfies
\begin{equation}
 \tilde{\omega}_n^+ \leq \pi^2 s^2 n^2 + \left(s^{-2} + 2 + 4 \tilde{\lambda}^2 s^2 \right)
\end{equation}
and so
\begin{equation}
\label{equation number of eigenvalues lower bound}
  N(\tilde{\omega}_{\text{max}}) \geq N_{R_2}(\tilde{\omega}_{\text{max}}) \geq 
  \left\lfloor \frac{1}{s\pi \sqrt{2}}\sqrt{ \tilde{\omega}_{\text{max}} - s^{-2} - 2 - 4\tilde{\lambda} s^2 } \right\rfloor
\end{equation}
Thus, if we take $\tilde{\omega}_{\text{max}}$ sufficiently large, depending on $s$ and $\tilde{\lambda}$ but \emph{independent} of $\mu_\ell$, then there is at least one eigenvalue $\tilde{\omega}$ for the linear problem \eqref{equation linear eigenvalue} satisfying $\tilde{\omega} \leq \tilde{\omega}_{\text{max}}$. Moreover, we see that the number of such eigenvalues grows at least as fast as $\sqrt{\tilde{\omega}_{\text{max}}}$ for large values of $\tilde{\omega}_{\text{max}}$.

We also want to show that an eigenvalue $\tilde{\omega}$ can be found satisfying $\tilde{\omega} > 4$. This follows from a very similar argument, although in this case we should count the number of eigenvalues to the associated problem with \emph{Neumann} boundary conditions, and where the potential is replaced by its \emph{minimum} value in the interval $\frac{1}{2} \mu_\ell^{-\frac{1}{2}} \leq y \leq \mu_\ell^{-\frac{1}{2}}$. See \cite{Holzegel2014} and \cite{Keir2016} for additional details.

In view of \eqref{equation potential in region y sim muell -1/2}, for sufficiently large $\mu_\ell$, in the region $\frac{1}{2} \mu_\ell^{-\frac{1}{2}} \leq y \leq \mu_\ell^{-\frac{1}{2}}$ we have $V_3 \geq \frac{1}{2} \left( s^{-4} + 2s^{-2} + 4\tilde{\lambda}^2 \right)$. Then, following the arguments given in \cite{Holzegel2014} (see also \cite{Keir2016}), we can see that the number of eigenvalues below $\frac{1}{2} \tilde{\omega}_{\text{max}}$ obeys the bound
\begin{equation}
	N\left( \frac{1}{2} \tilde{\omega}_{\text{max}} \right)
	\leq
	1
	+
	\left\lfloor \frac{1}{s\pi \sqrt{2}}\sqrt{ \frac{1}{2} \tilde{\omega}_{\text{max}} - \frac{1}{4} \left( s^{-2} + 2 + 4\tilde{\lambda} s^2 \right) } \right\rfloor
\end{equation}
Combining this bound with the bound \eqref{equation number of eigenvalues lower bound}, we see that, if $\tilde{\omega}_{\text{max}}$ is chosen sufficiently large, then for all sufficiently large $\mu_\ell$ there is at least one eigenvalue in the range $\frac{1}{2} \tilde{\omega}_{\text{max}} \leq \omega_{\text{max}} \leq \tilde{\omega}_{\text{max}}$.

\end{proof}

We need to connect this linear eigenvalue problem with the original nonlinear eigenvalue problem. In other words, we need to replace $\tilde{\omega}_0$, (where it appears in the effective potential) with $\tilde{\omega}$, and we also need to reintroduce the dependence of $\mu_\ell$ on $\tilde{\omega}$, by choosing $\mu_\ell$ to be a solution to the angular eigenvalue problem \eqref{equation defining mu ell}, rather than an arbitrary constant as in the linear eigenvalue case.

To connect the calculation done above with the original, nonlinear eigenvalue problem we introduce a parameter $b \in [0,1]$ and define the operator
\begin{equation}
 Q(b, \mu_\ell, \tilde{\omega}, \tilde{\lambda})(w) := 
 -\frac{\upd^2}{\upd y^2}w + V_3(y; \mu_\ell, b \tilde{\omega}, \tilde{\lambda}) w - 2s^{-2}\mu_\ell \tilde{\omega} w
\end{equation}
which also acts on functions in $H^1([0, R\mu_\ell])$, and where $\mu_\ell(b, \tilde{\omega}, \tilde{\lambda})$ is the square root of the $\ell$-th eigenvalue associated with the problem
\begin{equation}
\label{equation angular eigenvalue problem with b}
 \begin{split}
  &-\frac{1}{\sin 2\theta} \frac{\upd}{\upd \theta} \left( \sin 2\theta \frac{\upd u_{\theta, b, \ell}}{\upd \theta} \right) + \left( \frac{\ell^2}{\sin^2 \theta} + b^2(\tilde{\lambda}^2 - \tilde{\omega}^2)\frac{a^2}{(R_z)^2} \cos^2 \theta \right)u_{\theta, b, \ell} = (\mu_\ell)^2 u_{\theta, b,\ell} \\
  & u_{\theta,b,\ell} \big|_{\theta = 0} = u_{\theta,b,\ell}\big|_{\theta = \frac{\pi}{2}} = 0
 \end{split}
\end{equation}

Regarding this problem, we will prove the following proposition:
\begin{proposition}
	Let $b \in [0,1]$ and $\lambda \in \mathbb{R}$, and let $\mu_\ell = \mu_\ell(b, \tilde{\omega}, \tilde{\lambda})$ be the $\ell$-th solution to the eigenvalue problem \eqref{equation angular eigenvalue problem with b}.
	
	Suppose that, if $u_r$ is a function satisfying
	\begin{equation}
	\begin{split}
		Q( b, \mu_\ell, \tilde{\omega}, \tilde{\lambda} ) (u_r) = 0
		\\
		||u_r||_{L^2[0, R\mu_\ell]} = 1
	\end{split}
	\end{equation}
	then $u_r$ also satisfies the bound
	\begin{equation}
	\label{equation Agmon bootstrap}
		\int_{(\mu_\ell)^{-\frac{1}{2} + \epsilon}}^{R \mu_\ell} |u_r|^2 \upd y
		\lesssim 
		e^{-\delta|\tilde{\omega}| (\mu_\ell)^{\frac{1}{2} + \epsilon}} ||u_r||_{L^2[0, R\mu_\ell]}^2
	\end{equation}

	Then, if $\ell$ is sufficiently large, we have
	\begin{equation*}
		\mu_\ell = \ell + 1 + \mathcal{O}(\ell^{-1})
	\end{equation*}
	Moreover, we can pick some $\tilde{\omega}_{\text{max}}$, \emph{independent of $\ell$}, such that there exists some $\tilde{\omega} \geq 2$ satisfying $\tilde{\omega} < 2\tilde{\omega}_{\text{max}}$ and also such that the operator $Q(b, \mu_\ell, \tilde{\omega}, \tilde{\lambda})$ has zero as an eigenvalue.
	
\end{proposition}

The important thing is that $\tilde{\omega}$ can be bounded \emph{uniformly} in $\mu_\ell$, so that $\tilde{\omega}$ does not scale with $\mu_\ell$ (or, equivalently, with $\ell$). We shall also need to show that $\mu_\ell$ can be chosen to be arbitrarily large by choosing $\ell$ sufficiently large, and in addition we shall need to justify the bound \eqref{equation bootstrap 2}.

The extra assumption in this proposition, yielding the bound \eqref{equation Agmon bootstrap} on normalised solutions to the eigenvalue problem, will be proven below, in subsection \ref{subsection cut off}.

\begin{proof}

For $b = 0$, we have\footnote{The eigenfunctions in question are given by $u_{\theta,0, \ell} \propto \sin^\ell \theta$. } $(\mu_\ell)^2 = \ell(\ell+2)$, which in particular means that $\mu_\ell > \ell$, so that we can take $\mu_\ell$ to be arbitrarily large by taking $\ell$ to be large. In addition, this explicit equation for $\mu_\ell$ makes the bound \eqref{equation bootstrap 1} easy to prove in this case. The above calculations (with $\omega_0 = 0$ and $\tilde{\lambda}$ some fixed constant) then show that, for sufficiently large $\ell$, we can find some $\tilde{\omega}$ such that $Q(0, \mu_\ell, \tilde{\omega}, \tilde{\lambda})$ has a zero eigenvalue. Moreover, using proposition \ref{proposition linear eigenvalue problem} we see that, if we take some $\tilde{\omega}_{\text{max}}$ sufficiently large (but, importantly, independent of $\ell$), then for all sufficiently large $\ell$ it is possible to choose $\tilde{\omega}$ to satisfy the bound $\tilde{\omega} \leq \tilde{\omega}_{\text{max}}$.

The $N$-th eigenvalue of $Q$ is given by $\Lambda_N$ where
\begin{equation}
  \Lambda_N(b, \mu_\ell, \tilde{\omega}, \tilde{\lambda}) := \int_0^{R \mu_\ell} \left( \left(\frac{\upd u_r}{\upd y^2} \right)^2 + V_3(y; \mu_\ell, b \tilde{\omega}, \tilde{\lambda}) (u_r)^2 - 2s^{-2}\ell \tilde{\omega} (u_r)^2 \right) \upd y
\end{equation}
where $u_r = u_r(y; N, b, \mu_\ell(b), \tilde{\omega}, \tilde{\lambda})$ is the corresponding eigenfunction, normalised such that $\int_0^{R\mu_\ell} |u_r|^2\upd y = 1$. 


The angular eigenvalues $\mu_\ell$ are themselves functions of $b$, $\tilde{\omega}$ and $\tilde{\lambda}$. In fact, we have
\begin{equation}
\label{equation angular eigenvalue}
  \begin{split}
    \left(\mu_\ell(b, \tilde{\omega}, \tilde{\lambda})\right)^2 := \int_0^{\frac{\pi}{2}} \left(   \left( \frac{\upd u_\theta}{\upd \theta} \right)^2 + \left( \frac{\ell^2}{\sin^2 \theta} + b^2(\tilde{\lambda}^2 - \tilde{\omega}^2)\frac{a^2}{(R_z)^2} \cos^2 \theta \right) \left(u_{\theta} \right)^2  \right)\sin 2\theta \, \upd \theta
  \end{split}
\end{equation}
where the eigenfunction $u_\theta$ is normalised by
\begin{equation}
 \int_0^{\frac{\pi}{2}} (u_\theta)^2 \sin 2\theta \, \upd\theta = 1
\end{equation}

The results proved in \cite{Moeller} -- in particular, theorem 4.1, corollary 4.2 and theorem 7.1 -- show that, for the kinds of eigenvalue problems considered here, the eigenvalues are differentiable functions of the parameters. Moreover, the derivatives can be computed by applying the derivatives of the operators to the associated eigenfunctions. From this it follows that
\begin{equation}
\label{equation derivatives of mu}
  \begin{split}
    \left| \frac{\partial}{\partial \tilde{\omega}} \mu_\ell \right| &\leq b^2 \tilde{\omega} \frac{a^2}{(R_z)^2} |\mu_\ell|^{-1}\\
    \left| \frac{\partial}{\partial b} \mu_\ell \right| &\leq b \left( \tilde{\lambda}^2 + \tilde{\omega}^2\right) \frac{a^2}{(R_z)^2}|\mu_\ell|^{-1}
  \end{split}
\end{equation}

Combining the above few statements, we find that
\begin{equation}
\label{equation dLambda domega}
  \begin{split}
    \frac{\partial}{\partial \tilde{\omega}} \Lambda_N(b, \mu_\ell, \tilde{\omega}, \tilde{\lambda}) &= \int_0^{R \mu_\ell} \Bigg( \frac{\partial \mu_\ell}{\partial \tilde{\omega}} \frac{\partial}{\partial \mu_\ell}V_3(y;\mu_\ell,b\tilde{\omega}, \tilde{\lambda}) + b\frac{\partial}{\partial (b\tilde{\omega})} V_3(y;\tilde{\mu}_\ell,b\tilde{\omega}, \tilde{\lambda}) -2s^{-2} \ell \Bigg)|u_r|^2 \, \upd y
  \end{split}
\end{equation}
where we are considering $\ell$ to be some fixed integer. If we can show that the expression above is nonzero then we can appeal to the implicit function theorem to find that, for $b$ sufficiently small, we may find some $\omega$ such that $\Lambda_{N}(b, \mu_\ell, \tilde{\omega}, \tilde{\lambda}) = 0$. We shall prove this below, for sufficiently large $\mu_\ell$. The idea is to show that the first two terms in the integrand above are $\mathcal{O}(1)$, while the final term is exactly $-2s^{-2} \ell$ due to the normalisation condition on $u_r$.

First, we pick some very large constant $\tilde{\omega}_{\text{max}}$. Let $b_{\text{max}}$ be the largest value of $b \in [0,1]$ such that the following holds:
\renewcommand{\labelenumi}{(\Roman{enumi}.)}
\begin{enumerate}
 \item For all sufficiently large $\mu_\ell$, there exists a value of $\tilde{\omega}$ satisfying $|\tilde{\omega}| \leq 2 \tilde{\omega}_\text{max}$ and such that $\Lambda_N(b, \mu_\ell, \tilde{\omega}, \tilde{\lambda}) = 0$
 \item $\ell$ can be expressed in terms of $\mu_\ell$ as
  \begin{equation}
    \ell = \mu_\ell - 1 + \mathcal{O}\left( (\mu_\ell)^{-1} \right) \\
  \end{equation}
\end{enumerate}
Note that condition $(\mathrm{II}.)$ ensures that $\mu_\ell$ can be taken to be arbitrarily large (by taking $\ell$ large). As noted above, at $b = 0$, both conditions are satisfied (as long as $\tilde{\omega}_{\text{max}}$ is sufficiently large) and so $b_{\text{max}} \geq 0$.

We will now employ a bootstrap argument: assuming that conditions $(\mathrm{I}.)$ and $(\mathrm{II}.)$ hold for all $b \in [0, b_{\text{max}}]$ we will show that they actually hold for $b \in [0, \min\{ b_{\text{max}} +\epsilon, 1 \}]$. From this it follows immediately that $b_{\text{max}} = 1$.

From the explicit form of the potential $V_3$ we find
\begin{equation}
\label{equation derivatives of V3}
\begin{split}
	\frac{\partial V_3}{\partial \mu_\ell}  
	&=
	\frac{2 \mu_\ell y^2}{(y^2 + s^2)^2} 
	+ \frac{2 s^2 \mu_\ell}{(y^2 + s^2)^2}
	\\ \\
	b\frac{\partial V_3}{\partial(b\tilde{\omega})} 
	&=
	2bs^{-2} \ell \left( 1 - \frac{s^4}{(y^2 + s^2)^2} \right)
	\\
	&\phantom{=}
	-\frac{2 b^2\tilde{\omega}}{(y^2 + s^2)^2}\left( 
		y^4 
		+ \left(s^2 + \frac{Q_1 + Q_2}{(R_z)^2} \right) y^2
		+ s^2 \left( 1 + \frac{Q_1 + Q_2}{(R_z)^2} \right)
	\right)
 \end{split}
\end{equation}
For $b \leq b_{\text{max}}$ (where, in particular, $\tilde{\omega} = \mathcal{O}(1)$) we have $\frac{\partial \mu_\ell}{\partial \tilde{\omega}} = \mathcal{O}\left((\mu_\ell)^{-1}\right)$ so we only need to show that $\frac{\partial V_3}{\partial \mu_\ell}$ is $\mathcal{O}(\mu_\ell)$, which is clear from the explicit formula above.

On the other hand, all of the terms in the equation above for $b\frac{\partial V_3}{\partial(b\tilde{\omega})}$ are clearly $\mathcal{O}(1)$ except for the first one, which appears to be $\mathcal{O}(\mu_\ell)$. In fact, below we will show that it is almost $\mathcal{O}(1)$.

Equation \eqref{equation Agmon bootstrap} (which we will justify in the following section) implies that the contribution to the integrals in equation \eqref{equation dLambda domega} from the region $y \geq C (\mu_\ell)^{-\frac{1}{2} + \epsilon}$ is exponentially suppressed in $\mu_\ell$; in particular, it is $\mathcal{O}(1)$ at large $\mu_\ell$. In the complement of this region, we have the bound (for sufficiently large $\mu_\ell$):
\begin{equation}
  2bs^{-2} \ell \left( 1 - \frac{s^4}{(y^2 + s^2)^2} \right) \lesssim C^2 s^{-2} \ell (\mu_\ell)^{-1+2\epsilon} \lesssim C^2 s^{-2}(\mu_\ell)^{2\epsilon} \quad \text{for }y \leq C(\mu_\ell)^{-\frac{1}{2} + \epsilon}
\end{equation}
where we have also restricted to $b \leq b_{\text{max}}$.

Now, combining the above calculations, we find that, for $b \leq b_{\text{max}}$ 
\begin{equation}
 \frac{\partial}{\partial \tilde{\omega}} \Lambda_N(b, \mu_\ell, \tilde{\omega}, \tilde{\lambda}) = -2s^{-2}\mu_\ell + \mathcal{O}\left( (\mu_\ell)^\epsilon \right)
\end{equation}
where we may choose $\epsilon$ to be arbitrarily small. Hence, for sufficiently large $\mu_\ell$, this is bounded away from zero. In fact, for all sufficiently large $\mu_\ell$ we have the bound
\begin{equation}
\label{equation dLambda domega final}
 \left| \frac{\partial}{\partial \tilde{\omega}} \Lambda_N(b, \mu_\ell, \tilde{\omega}, \tilde{\lambda}) \right| \geq  s^{-2}|\mu_\ell|
\end{equation}
Hence, we can use the implicit function theorem and conclude that, for some small $\epsilon$ and for all sufficiently large $\mu_\ell$, for all $b \in [0, b_{\text{max}} + \epsilon]$ there exists a value of $\tilde{\omega}$ such that
\begin{equation*}
 \Lambda_N(b, \mu_\ell, \tilde{\omega}, \tilde{\lambda}) = 0
\end{equation*}

By a series of almost identical calculations we can also show that, at least for $b \leq b_{\text{max}}$, we have
\begin{equation}
\label{equation dLambda db}
 \left| \frac{\partial}{\partial b} \Lambda_N(b, \mu_\ell, \tilde{\omega}, \tilde{\lambda}) \right| \lesssim 1
\end{equation}
Now, the implicit function theorem allows us to combine equations \eqref{equation dLambda domega final} and \eqref{equation dLambda db} to conclude that, for all $b \leq b_{\text{max}}$, for the choice of $\tilde{\omega}$ made above we have
\begin{equation}
 \left|\frac{\partial \tilde{\omega}}{\partial b} \right| \lesssim (\mu_\ell)^{-1}
\end{equation}
Moreover, recall that, at $b = 0$ we can choose $\tilde{\omega}$ to also satisfy
\begin{equation*}
 \tilde{\omega} \leq \tilde{\omega}_{\text{max}}
\end{equation*}
Hence we can conclude that, for all sufficiently large $\mu_\ell$, and for all $b \in [0, b_{\text{max}} + \epsilon]$ it is possible to find a value of $\tilde{\omega}$ such that $\Lambda_N = 0$ and moreover such that $\tilde{\omega} \leq 2\omega_{\text{max}}$. In other words, condition $(\mathrm{I}.)$ is actually satisfied for a slightly larger range of values of $b$.

Now, we only need to conclude that the same is true for condition $(\mathrm{II}.)$. Recall that, from equation \eqref{equation derivatives of mu} we have the bound
\begin{equation}
 \left| \frac{\partial}{\partial b} \mu_\ell \right| \leq b \left( \tilde{\lambda}^2 + \tilde{\omega}^2\right) \frac{a^2}{(R_z)^2}|\mu_\ell|^{-1}
\end{equation}
and so, for $b \leq b_{\text{max}}$ we have
\begin{equation}
 \left| \frac{\partial}{\partial b} \mu_\ell \right| \lesssim |\mu_\ell|^{-1}
\end{equation}
Also recall that, at $b = 0$ we have $\mu_\ell = \sqrt{\ell(\ell + 2)}$. Hence, integrating from $b = 0$ to $(b_{\text{max}} + \epsilon)$ we find that condition $(\mathrm{II}.)$ also holds for slightly larger values of $b$. But this contradicts the maximality of $b_{\text{max}}$, so we must be able to take $b_{\text{max}} = 1$.

\end{proof}

In short, we have proved that, in the two charge microstate geometries, if we choose some sufficiently large $\tilde{\omega}_{\text{max}}$, then for all sufficiently large $\ell$ there are solutions to the associated Dirichlet problem, with boundary conditions imposed at $y = R\mu_\ell$, and moreover these solutions have associated frequencies satisfying the bound $\tilde{\omega} \leq 2\omega_{\text{max}}$. Most importantly, this bound is \emph{uniform} in $\ell$, i.e.\ $\tilde{\omega}_{\text{max}}$ does not depend on $\ell$. Finally, we have shown that the associated angular eigenvalues $\mu_\ell$ also obey the bound $\mu_\ell \sim \ell$. The only part of the proof not yet justified is equation \eqref{equation Agmon bootstrap}, the proof of which is given in the next section.

\subsection{The cut-off and the associated error}
\label{subsection cut off}

We can extend the quasimodes constructed in the previous subsection to continuous function on the whole of spacetime by setting $u_r \equiv 0$ in the outer region $r \geq R \mu_\ell$, however, these functions will not even be $C^1$ in general. Instead, we will multiply $u_r$ by a smooth cut-off function $\chi_\ell(r)$, with support in the region $r \in [\frac{1}{2} R \mu_\ell, R \mu_\ell]$, such that $\chi_\ell(\frac{1}{2} R\mu_\ell) = 1$ and $\chi_\ell(R\mu_\ell) = 0$.

This means that the quasimodes no longer give a solution to the wave equation, but only an \emph{approximate} solution. To quantify the error produced by the cut-off, we rely on Agmon estimates as in \cite{Holzegel2014} and \cite{Keir2016}. We begin with a simple weighted energy-type estimate, which can be proven by integrating by parts, and appears as lemma 5.1 in \cite{Keir2016}:

\begin{proposition}[Exponentially weighted energy estimate]
\label{proposition exponentially weighted energy}
Let $r_2 \geq r_1$, let $h > 0$ be a positive constant and let $u$, $W$ and $D$ be smooth, real valued functions of $r \in [r_1, r_2]$, with $u(r_1) = u(r_2) = 0$. Then
\begin{equation}
 \begin{split}
  &\int_{y_1}^{y_2} \left( \left| \frac{\upd}{\upd y} \left( e^{D/h} u \right) \right|^2 + h^{-2}\left( W - \left(\frac{\upd D}{\upd r}\right)^2 \right)e^{2D/h}|u|^2 \right) \upd y \\
  &= \int_{r_1}^{r_2} \left( -\frac{\upd^2 u}{\upd y^2} + h^{-2}W u \right) u \, e^{2D/h} \upd y
 \end{split}
\end{equation}

\end{proposition}

Let $u_r$ be the $N$-th eigenfunction of $Q(b, \mu_\ell, \tilde{\omega}, \tilde{\lambda})$, with zero as the corresponding eigenvalue, and let $U_{\text{forbid}}$ be an open set such that
\begin{equation}
 (\mu_\ell)^{-2} V_3(y; b\mu_\ell, \tilde{\omega}, \tilde{\lambda}) - 2s^{-2}\tilde{\omega}\ell(\mu_\ell)^{-2} > \delta \text{\quad for \quad} r \in U_{\text{forbid}}
\end{equation}
for some small positive constant $\delta$ to be chosen later. We refer to $U_{\text{f}}$ as the ``forbidden region''. 

Similarly, we can define the ``classical region'': we first define
\begin{equation}
  \tilde{U}_{\text{classic}} := \{ y \, | \, (\mu_\ell)^{-2} V_3(y; b\mu_\ell, \tilde{\omega}, \tilde{\lambda}) - 2s^{-2}\tilde{\omega}\ell(\mu_\ell)^{-2} < 0 \}
\end{equation}
then we define the classical region, $U_{\text{classic}}$, as the connected component of $\tilde{U}_{\text{classic}}$ containing the smallest values of $y$.

 We define the Agmon distance between two points $y_1$ and $y_2$ as
\begin{equation}
 d_{\tilde{\omega}}(y_1, y_2) := \left| \int_{y_1}^{y_2} (\mu_\ell)^{-1} \sqrt{ \sup\{0, V_3(y; b\mu_\ell, \tilde{\omega}, \tilde{\lambda}) - 2s^{-2}\tilde{\omega}\ell \} } \, \upd y \right|
\end{equation}
which is easily seen to satisfy
\begin{equation}
 \left| \frac{\partial}{\partial y} d_{\tilde{\omega}}(y_1,y) \right| \leq (\mu_\ell)^{-1} \sqrt{\left| V_3(y; b\mu_\ell, \tilde{\omega}, \tilde{\lambda}) - 2s^{-2}\tilde{\omega}\ell \right|}
\end{equation}
We can define the distance to the ``classical region'' as
\begin{equation}
	d_{\tilde{\omega}}^{\text{classic}}(y) := \inf_{y_1 \leq y \, |\,  y_1 \in U_{\text{classic}}} d_{\tilde{\omega}}(y_1, y)
\end{equation}
where we also define $d_\omega^{\text{classic}}(y) := 0$ if there are no points $y_1 \leq y$ such that $y_1 \in U_{\text{classic}}$. Note that $d_\omega^{\text{classic}}(y)$ measures the distance from the point $y$ to the nearest point in the classical region ``to the left'' of $y$, i.e.\ to the nearest point in the classical region at a smaller radius than $y$.

We now pick any small $\delta > 0$ and perform the exponentially weighted energy estimate of proposition \ref{proposition exponentially weighted energy} with the following choices for the parameters and functions:
\begin{equation}
  \begin{split}
   u &= u_r \\
   h &= (\mu_\ell)^{-1} \\
   W &= (\mu_\ell)^{-2}\left(V_3(y; b\mu_\ell, \tilde{\omega}, \tilde{\lambda}) - 2s^{-2}\tilde{\omega}\ell\right)  \\
   D &= (1-\delta)d_{\tilde{\omega}}^{\text{classic}} \\
   y_1 &= 0 \\
   y_2 &= R \mu_\ell
  \end{split}
\end{equation}
This gives the estimate
\begin{equation}
\label{equation Agmon estimate 1}
	\begin{split}
		&\int_{U_{\text{forbid}}} \left| \frac{\partial}{\partial y} \left( e^{(1-\delta)\mu_\ell d_{\tilde{\omega}}^{\text{classic}}} u_r \right) \right|^2 \upd y \\
		&+ \int_{U_{\text{forbid}}} \left( V_3 - 2s^{-2}\tilde{\omega}\ell - (\mu_\ell)^2(1-\delta)^2 \left( \frac{\partial}{\partial y} d_{\tilde{\omega}}^{\text{classic}} \right)^2 \right)e^{2(1-\delta)\mu_\ell d_{\tilde{\omega}}^{\text{classic}}}|u_r|^2 \upd y \\
		&= -\int_{[0\, ,\,  R\mu_\ell] \setminus U_{\text{forbid}}} \left| \frac{\partial}{\partial y} \left( e^{(1-\delta)\mu_\ell d_{\tilde{\omega}}^{\text{classic}}} u_r \right) \right|^2 \upd y \\
		&\phantom{=} + \int_{[0\, ,\,  R\mu_\ell] \setminus U_{\text{forbid}}} \left( 2s^{-2}\tilde{\omega}\ell - V_3+ (\mu_\ell)^2(1-\delta)^2\left( \frac{\partial}{\partial y} d_{\tilde{\omega}}^{\text{classic}} \right)^2 \right)e^{2(1-\delta)\mu_\ell d_{\tilde{\omega}}^{\text{classic}}}|u_r|^2 \upd y \\
	\end{split}
\end{equation}
Now, in the region $U_{\text{forbid}}$ we have that
\begin{equation*}
	\left( \frac{\partial}{\partial y} d_{\tilde{\omega}}^{\text{classic}} \right)^2 = (\mu_\ell)^{-2}\left(V_3 - 2s^{-2}\tilde{\omega}\ell\right) > \delta
\end{equation*}
In particular, this means that for $\delta$ sufficiently small, in the forbidden region $U_{\text{forbid}}$, we have
\begin{equation*}
	\left(V_3 - 2s^{-2}\tilde{\omega}\ell\right) - (\mu_\ell)^2(1-\delta)^2 \left( \frac{\partial}{\partial y} d_{\tilde{\omega}}^{\text{classic}} \right)^2 \gtrsim (\mu_\ell)^2\delta^2
\end{equation*}
On the other hand, for sufficiently large $\ell$, $V_3$ satisfies the global inequality
\begin{equation*}
	V_3 \geq -\frac{1}{4}y^{-2}
\end{equation*}
which can be seen by the following calculation: first we note that $y^2 (y^2 + s^2)^2 V_3$ is a polynomial in $y$. Keeping only the leading order terms in each coefficient of this polynomial, we find that
\begin{equation}
	V_3
	\sim
	\frac{1}{y^2 (y^2 + s^2)^2} \left(
		2s^{-2} \mu_\ell \tilde{\omega} y^6
		+ (\mu_\ell)^2 y^4
		+ 2s^2 \mu_\ell y^2
		+ s^4 \left( \tilde{\lambda}^2 - \frac{1}{4} \right)
	\right)
\end{equation}
and so, for sufficiently large $\mu_\ell$, $V_3 \geq -\frac{1}{4}y^{-2}$.

We can also see that, for sufficiently large $\mu_\ell$, the classical region is disconnected from infinity, i.e.\ it is a finite interval. In view of the fact that $y^2 (y^2 + s^2)^2 \left( V_3 - 2s^{-2} \tilde{\omega} \ell \right)$ is a polynomial in $y$, keeping only the leading order terms in this polynomial we obtain
\begin{equation}
\label{equation V3 - eigenvalue}
	V_3 - 2s^{-2} \tilde{\omega} \ell
	\sim
	\frac{1}{y^2 (y^2 + s^2)^2} \left(
		-(\tilde{\omega}^2 - \tilde{\lambda}^2) y^6
		+ (\mu_\ell)^2 y^4
		- 2s^2 \mu_\ell (\tilde{\omega} - 1) y^2
		+ s^4 \left(\tilde{\lambda}^2 - \frac{1}{4} \right)
	\right)
\end{equation}
If $\tilde{\lambda} = 0$ then it is clear that $V_3 - 2s^{-2} \tilde{\omega} < 0$ for sufficiently small $y$. On the other hand, if $\tilde{\lambda} > 0$ and $\tilde{\omega} \geq 2$ then there is some region near the origin where the third term on the right hand side of equation \eqref{equation V3 - eigenvalue} is dominant, coinciding with the classical region. The fact that, for $\tilde{\lambda} > 0$, the classical region does not extend all the way to $y = 0$ is a manifestation of the ``angular momentum barrier'': recall that the coordinates $(y, z)$ parametrise a plane, with $y$ playing the role of the radial coordinate and $z$ the angular coordinate, and that $\tilde{\lambda}$ is associated with the angular momentum of the wave in the $z$ direction.

In the complement of the classical region we have (by definition)
\begin{equation*}
	V_3 \geq 2s^{-2}\tilde{\omega}\ell
\end{equation*}
while, inside the classical region we have $d_{\tilde{\omega}}^{\text{classic}} = 0$. We claim the following:
\begin{equation*}
 \int_{U_{\text{classic}}} \left( \left( \frac{\upd u_r}{\upd y} \right)^2 - \frac{1}{4} y^{-2} (u_r)^2 \right) \upd y \geq 0
\end{equation*}
which we prove below:
\begin{equation*}
 \begin{split}
  \int_{U_{\text{classic}}} \frac{1}{2}y^{-2}(u_r)^2 \upd y &= \int_{U_{\text{classic}}} -\frac{1}{2} \partial_y (y^{-1}) (u_r)^2 \upd y \\
  &= \left.\left( -\frac{1}{2} y^{-1} (u_r)^2 \right)\right|_{\partial U_{\text{classic}}} + \int_{U_{\text{classic}}} y^{-1} (u_r)(\partial_y u_r) \upd y
 \end{split}
\end{equation*}
Now, the first term on the right hand side is non-positive, since $u_r \in H^1_0([0, R\mu_\ell])$ so the contribution from the boundary at $y = 0$ vanishes, while the contribution from the other boundary is negative. Meanwhile, the other term can be bounded as
\begin{equation*}
 \int_{U_{\text{classic}}} y^{-1} (u_r)(\partial_y u_r) \upd y \leq \int_{U_{\text{classic}}} \left( \frac{1}{4} y^{-2} (u_r)^2 + (\partial_y u_r)^2 \right) \upd y
\end{equation*}
proving the claim.


Combining the above few estimates with equation \eqref{equation Agmon estimate 1} we find that, for all sufficiently small $\delta$,
\begin{equation}
\label{equation Agmon estimate 2}
	\begin{split}
		&\int_{U_{\text{forbid}}} \left| \frac{\partial}{\partial y} \left( e^{(1-\delta)\mu_\ell d_{\tilde{\omega}}^{\text{classic}}} u_r \right) \right|^2 \upd y 
		+ \delta^2 (\mu_\ell)^2 \int_{U_{\text{forbid}}} e^{2(1-\delta)\mu_\ell d_{\tilde{\omega}}^{\text{classic}}}|u_r|^2 \upd y \\
		&\leq 2s^{-2}\tilde{\omega}\ell e^{2(1-\delta) \mu_\ell a_{\tilde{\omega}}(\delta)} \int_{[0\, ,\,  R\mu_\ell] \setminus U_{\text{forbid}}} |u_r|^2 \upd y \\
		&\leq 2s^{-1}\tilde{\omega}\ell e^{2(1-\delta) \mu_\ell a_{\tilde{\omega}}(\delta)} ||u_r||_{L^2[0, R\mu_\ell]}^2 \\
	\end{split}
\end{equation}
where $a_{\tilde{\omega}}(\delta)$ is defined by
\begin{equation}
	a_{\tilde{\omega}}(\delta) := \sup_{y \in [0\, ,\,  R\mu_\ell] \setminus U_{\text{forbid}}} \; d_{\tilde{\omega}}^{\text{classic}}(y)
\end{equation}

Now, suppose that there is some subset of the forbidden region, say $U \subset U_{\text{forbid}}$, and within $U$ we have some lower bound on $d_{\tilde{\omega}}^{\text{classic}}$ which is \emph{uniform} in $\delta$. Then we have the following lemma:

\begin{lemma}
  \label{lemma Agmon estimate}
	Let $U \subset U_{\text{forbid}}$ be some subset of the forbidden region such that the Agmon distance to the classical region satisfies some lower bound:
	\begin{equation}
		d_{\tilde{\omega}}^{\text{classic}} \geq d_U^{\text{classic}}(\mu_\ell)
	\end{equation}
	\emph{independent} of $\delta$. Moreover, let $U$ be a region such that $V_3(y; b\mu_\ell, \tilde{\omega}, \tilde{\lambda})$ is bounded for all $y \in U$ (i.e.\ $0 \notin U$).

	Then for any constant $\epsilon_0 > 0$ there exists some positive constant $C(\epsilon_0, \tilde{\omega}_{\text{max}})$ such that, for all sufficiently large $\mu_\ell$ we have the bound
	\begin{equation}
		||u_r||^2_{H^1[U]} \leq C(\epsilon_0, \tilde{\omega}_{\text{max}}) (\mu_\ell) e^{\mu_\ell\left(\epsilon_0 - d_U^{\text{classic}}\right)}||u_r||^2_{L^2[0, R\mu_\ell]}
	\end{equation}
	and indeed, for all $k \geq 1$ there exists some constant $C_k$ such that
	\begin{equation}
		||u_r||^2_{H^k[U]} \leq C_k(\epsilon_0, \tilde{\omega}_{\text{max}}) (\mu_\ell)^{4k - 3} e^{\mu_\ell\left(\epsilon_0 - d_U^{\text{classic}}\right)}||u_r||^2_{L^2[0, R\mu_\ell]}
	\end{equation}
	
\end{lemma}
\begin{proof}
We start by considering the second integral on the left hand side of equation \eqref{equation Agmon estimate 2}, neglecting for the moment the first term.  Now, $U \subset U_{\text{forbid}}$, and within $U$ we have the $\mu_\ell$-dependent (but $\delta$-independent) lower bound $d_\omega^{\text{classic}} \geq d_U^{\text{classic}}(\mu_\ell)$. Hence, for all sufficiently small $\delta$, for all sufficiently large $\mu_\ell$ we have
\begin{equation}
	\delta^2 (\mu_\ell)^2e^{\mu_\ell d_U^{\text{classic}}(\mu_\ell)} \int_{U} |u_r|^2 \upd y
	\lesssim  \tilde{\omega}_{\text{max}}\mu_\ell e^{2 \mu_\ell a_{\tilde{\omega}}(\delta)} ||u_r||_{L^2[0, R\mu_\ell]}^2 \\
\end{equation}
Now, we note that $a_{\tilde{\omega}}(\delta) \rightarrow 0$ as $\delta \rightarrow 0$, uniformly in $\tilde{\omega}$. Hence, we may pick any $\epsilon_0 > 0$ and we can then find a positive constant $C(\epsilon_0, \mathcal{E}, C_1)$ such that the following inequality holds:
\begin{equation}
	||u_r||^2_{L^2[U]} \leq C(\epsilon_0, \tilde{\omega}_{\text{max}}) (\mu_\ell)^{-1} e^{\mu_\ell\left(\epsilon_0 - d_U^{\text{classic}}\right)}||u_r||^2_{L^2[0, R\mu_\ell]}
\end{equation}
which tells us that, as long as $d_U^{\text{classic}}$ is bounded away from zero, then the $L^2$ norm of $u_r$ is exponentially suppressed in the region $U$ at large $\mu_\ell$.

Now we return to the first term on the left hand side of equation \eqref{equation Agmon estimate 2}. We can estimate this as follows:
\begin{equation}
\label{equation Agmon estimate 3}
	\begin{split}
		&\int_{U_{\text{forbid}}} \left| \frac{\partial}{\partial y} \left( e^{(1-\delta)\mu_\ell d_{\tilde{\omega}}^{\text{classic}}} u_r \right) \right|^2 \upd y 
		\geq \int_U \left| \frac{\partial}{\partial y} \left( e^{(1-\delta)\mu_\ell d_{\tilde{\omega}}^{\text{classic}}} u_r \right) \right|^2 \upd y  \\
		&= \int_U e^{2(1-\delta)\mu_\ell d_{\tilde{\omega}}^{\text{classic}}} \bigg(  \left( \partial_y u_r \right)^2 + 2(1-\delta)\mu_\ell (\partial_y d_{\tilde{\omega}}^{\text{classic}})(\partial_y u_r)u_r  \\
		&\phantom{= \int_U e^{2(1-\delta)\mu_\ell d_{\tilde{\omega}}^{\text{classic}}} \bigg(}
		+ (1-\delta)^2(\mu_\ell)^2 (\partial_y d_{\tilde{\omega}}^{\text{classic}})^2 (u_r)^2  \bigg) \ \upd y
	\end{split}
\end{equation}
Both the first and third term on the right hand side are positive, so we can drop the third term. The second term can be bounded as follows:
\begin{equation}
	2(1-\delta)\mu_\ell (\partial_y d_{\tilde{\omega}}^{\text{classic}})(\partial_y u_r)u_r
	\lesssim \alpha \left( \sup_{U}\left( V_3 - 2s^{-2}\tilde{\omega}\ell \right) \right)^2 |u_r|^2
	+ (\alpha)^{-1} |\partial_y u_r|^2
\end{equation}
for any positive constant $\alpha > 0$. If we take $\alpha$ sufficiently large then we can absorb the term involving $|\partial_y u_r|^2$ by the first term on the right hand side of equation \eqref{equation Agmon estimate 3}, leading us to the estimate
\begin{equation}
	\begin{split}
		\int_U e^{2(1-\delta)\mu_\ell d_{\tilde{\omega}}^{\text{classic}}} |\partial_y u_r|^2 \upd y
		&\lesssim
		\tilde{\omega}_{\text{max}}\mu_\ell e^{2 \mu_\ell a_{\tilde{\omega}}(\delta)} ||u_r||_{L^2[0, R\mu_\ell]}^2 \\
		&\phantom{\lesssim} + \left( \sup_{U}\left( V_3 - 2s^{-2}\tilde{\omega}\ell \right) \right)^2 \int_U e^{2(1-\delta)\mu_\ell d_{\tilde{\omega}}^{\text{classic}}}|u_r|^2 \upd y
	\end{split}
\end{equation}
But we have $V_3 \lesssim (\mu_\ell)^2$, so the second term on the right hand can be bounded by the earlier calculations, and so we have
\begin{equation}
	||\partial_y u_r||^2_{L^2[U]} \lesssim C(\epsilon_0, \tilde{\omega}_{\text{max}}) (\mu_\ell) e^{\mu_\ell\left(\epsilon_0 - d_U^{\text{classic}}\right)}||u_r||^2_{L^2[0, R\mu_\ell]}
\end{equation} 
completing the proof of the first part of the lemma.

To prove the last part of the lemma, we simply note that the eigenvalue equation \eqref{equation nonlinear eigenvalue} satisfied by $\Psi_{\ell, N}$ allows us to bound
\begin{equation}
	||\partial_y^2 u_r||^2_{L^2(U)} \lesssim (\tilde{\omega}_{\text{max}})^2 (\mu_{\ell})^4 ||u_r||^2_{L^2(U)}
\end{equation}
and higher derivatives can be treated similarly.
\end{proof}


By choosing $U$ to be the region $y \gtrsim (\mu_\ell)^{-\frac{1}{2} + \epsilon}$ we can prove the bound in equation \eqref{equation Agmon bootstrap}. In fact, expanding in powers of $y$ the effective potential is given by
\begin{equation}
 \begin{split}
  V_3(y; \mu_\ell, b\tilde{\omega}, \tilde{\lambda}) - 2s^{-2}\ell\tilde{\omega}
  &= \left( \tilde{\lambda}^2 - \frac{1}{4} \right)y^{-2} \\
  &\phantom{=} + \bigg( 2s^{-2}\tilde{\lambda}^2 + s^{-2}\left( (\mu_\ell)^2 - \ell^2 \right) - s^{-2}\left( b^2 \tilde{\omega}^2 - \tilde{\lambda}^2 \right)\left( 1 + \frac{Q_1 + Q_2}{(R_z)^2} \right) \\
  &\phantom{= + \bigg( } + 2s^{-2} - 2s^{-2}\tilde{\omega}\ell \bigg) \\
  &\phantom{=} + \left( (\mu_\ell)^2 s^{-4} + s^{-4}\left( b^2 \tilde{\omega}^2 - \tilde{\lambda}^2 \right) \left( s^2 + \frac{Q_1 + Q_2}{(R_z)^2} \right) - 3s^{-4} \right)y^2 \\
  &\phantom{=} + \mathcal{O}\left( (\mu_\ell)^2 y^4 \right)
 \end{split}
\end{equation}

Note that the first term is $\mathcal{O}(1 - 2\epsilon)$. Since the leading order term is $\mathcal{O}(1 + 2\epsilon)$, this first term can only be neglected if $\epsilon > 0$. In other words, the following estimates are not uniform in $\epsilon$.

Hence, in the region $y \geq (\mu_\ell)^{-\frac{1}{2} + \epsilon}$ we have 
\begin{equation}
 V_3(y; \mu_\ell, b\tilde{\omega}, \tilde{\lambda}) - 2s^{-2}\ell\tilde{\omega}
 \gtrsim s^{-4} (\mu_\ell)^{1 + 2\epsilon}
\end{equation}
for sufficiently large $\mu_\ell$. Moreover, this bound holds for all values of $b \in [0,1]$. This leads directly to the bound in equation \eqref{equation Agmon bootstrap}.

Additionally, we can choose $U$ to be the region $\frac{1}{2}R\mu_\ell < y < R\mu_\ell$, coinciding with $\text{support}(\chi_\ell)$. Then, (setting $b = 1$) for $\ell$ sufficiently large, in view of \eqref{equation V3 - eigenvalue} we see that, to leading order in $\mu_\ell$ (using the fact that $y = \mathcal{O}(\mu_\ell)$)
\begin{equation*}
	V_3(y; \mu_\ell, \tilde{\omega}, \tilde{\lambda}) 
	- 2s^{-2}\ell \tilde{\omega}
	=
	y^{-2} (\mu_\ell)^2 - (\tilde{\omega}^2 - \tilde{\lambda}^2)
\end{equation*}

Hence there is some $y_1$ (independent of $\ell$) such that, for all sufficiently large $\ell$,
\begin{equation}
 (\mu_\ell)^{-2} \left( V_3(y; \mu_\ell, \tilde{\omega}, \tilde{\lambda}) - 2s^{-2}\ell \tilde{\omega} \right) \gtrsim  y^{-2} - (\tilde{\omega}^2 - \tilde{\lambda}^2) (\mu_\ell)^{-2} \text{\quad for \quad} y > y_1
\end{equation}
We now pick $R$ sufficiently small compared to $\sqrt{\tilde{\omega}^2 - \tilde{\lambda}^2}$, so that the right hand side of this inequality is positive. Note that $\sqrt{\tilde{\omega}^2 - \tilde{\lambda}^2} \leq \tilde{\omega}_{\text{max}}$, so we can choose $R$ independent of $\ell$. Hence, with this choice of $U$ we have the lower bound on $a_{\omega(U)}$:
\begin{equation}
 a_\omega(U) \gtrsim \mu_\ell \int_{y_1}^{\frac{1}{2}R\mu_\ell} \sqrt{ y^{-2} - (\mu_\ell)^{-2}(\tilde{\omega}^2 - \tilde{\lambda}^2)} \upd y \gtrsim \ell \log \ell
\end{equation}
Substituting this into lemma (\ref{lemma Agmon estimate}), and choosing \emph{any} value for $\epsilon_0$ (say $\epsilon_0 = 1$) we find that the size of the solution $u_r$ is super-exponentially suppressed in $\ell$ in the region $U$, i.e. there are positive constants $C_k > 0$ such that
\begin{equation}
 ||u_r||_{H^k}^2 \lesssim C_k e^{-C_k \ell \log \ell}||u_r||^2_{L^2[0, R\mu_\ell]}
\end{equation}

Now, we have that
\begin{equation}
 \Box_g (\chi u) = \chi \Box_g u + 2 (g^{-1})^{\mu\nu}(\partial_\mu \chi)(\partial_\nu u) + u (\Box_g \chi)
\end{equation}
so, if $u$ satisfies the wave equation $\Box_g u = 0$ and the first two derivatives of $\chi$ are supported in $U$ and bounded in $L^\infty$ by, say, $\chi_0$, then we have
\begin{equation}
 || \Box_g (\chi u) ||_{H^k[0, R \mu_\ell]} \lesssim \chi_0 ||u||_{H^{k+1}[U]}
\end{equation}
In particular, we conclude that, by setting
\begin{equation}
 u_\ell := \chi_{\ell}(r) e^{-i\omega t + i\lambda z - i\ell \phi} u_r(r) u_{\theta,\ell}(\theta)
\end{equation}
then $u$ is an approximate solution to the wave equation, with an error quantified by the estimate
\begin{equation}
\label{equation error}
 || \Box_g u_\ell ||^2_{H^k[0, R\mu_\ell]} \lesssim (\chi_0)^2 e^{-C_k \ell \log \ell}||u||^2_{L^2[0, R\mu_\ell]}
\end{equation}

\subsection{Bounding the uniform decay rate}

We can use the bound in the previous section to show that, if we take initial data which coincides with the initial data for the quasimodes, then the corresponding solutions remain close to the quasimodes for a long period of time. To show this, we need the following curved space version of Duhamel's principle:

\begin{proposition}[Duhamel's principle]
\label{proposition Duhamel}
Let $x \in \Sigma_t$ label a point in $\Sigma_t$.

Let $P(\tau, x; s)(f_1, f_2)$ be the solution to the homogeneous problem
\begin{equation}
 \begin{split}
  \Box_g P(\tau, x; s)(f_2, f_2) &= 0 \\
  P(t, x; s)(f_1, f_2)\big|_{\Sigma_s} &= f_1 \\
  \left.\frac{\partial}{\partial t}\right|_{x} P(\tau, x; s)(f_1, f_2) \big|_{\Sigma_s} &= f_2
 \end{split}
\end{equation}
i.e.\ $P(t, x; s)(f_1, f_2)$ is the solution to the homogeneous wave equation at the point $(t, x)$ with initial data $(f_1, f_2)$ posed at time $t = s$.

Then the solution to the problem
\begin{equation}
 \begin{split}
  \Box_g u &= F(t, x) \\
  u\big|_{\Sigma_0} &= u_0 \\
  \left.\frac{\partial}{\partial t}\right|_{x} u \big|_{\Sigma_s} &= u_1
 \end{split}
\end{equation}
for $F$ some bounded function, is given by
\begin{equation}
 u(t, x) = P(t, x; 0)(u_0, u_1) + \int_0^t P(t, x; s)\left(0, (g^{00})^{-1} F(s, x) \right) \upd s
\end{equation}

\end{proposition}

\begin{proof}
 Let $x^i$ be local coordinates on $\Sigma_t$. Defining $u$ as in the proposition, we find that
\begin{equation}
  \begin{split}
    \Box_g u(t, x) &= g^{00}\partial_t \left( P(t, x; t)\left(0, (g^{00})^{-1} F(s, x) \right) \right)
    + 2 g^{0i} \partial_i \left( P(t, x; t)\left(0, (g^{00})^{-1} F(s, x) \right) \right) \\
    &\phantom{=} + g^{\mu \nu} \Gamma^0_{\mu\nu} P(t, x; t)\left(0, (g^{00})^{-1} F(s, x) \right)
    + \int_0^t \Box_g P(t, s)\left( 0, (g^{00})^{-1} F(s, x) \right)\upd s \\
    &= F(s, x)
  \end{split}
\end{equation}
Moreover, $u$ satisfies the correct initial conditions.
\end{proof}

We can now use the results established above to provide counterexamples to possible uniform decay statements, which is the main theorem in this section:

\begin{theorem}
 Let $u$ denote a solution to the wave equation $\Box_g u = 0$ on a two charge microstate geometry. Let $U$ be any open set containing the submanifold $r = 0$.

Then for all $k \geq 1$ there exist positive constants $C_k$ and $c_k$ such that
\begin{equation}
 \limsup_{\tau \rightarrow \infty} \sup_{\substack{u,\  Tu \in \operatorname{SCH}(\Sigma_0) \\ ||u||_{H^{k}(\Sigma_0)} \neq 0} } \left( \frac{ ||u||_{H^{1}(U)} + ||T u||_{L^2(U)} }{ \sqrt{E_0^{(k)}(u)} } \right) \exp \left( (k-1) W_0 \left( \frac{1}{C_k} \log \left( \frac{t}{C_k} \right) \right) \right) \geq c_k
\end{equation}
where $W_0(\cdot)$ is the Lambert $W$ function, defined by
\begin{equation}
  \begin{split}
    W_0(x e^x) = x \\
    W_0 \geq -1 \\
  \end{split}
\end{equation}
and $\operatorname{SCH}(\Sigma_0)$ denotes the space of Schwartz functions on $\Sigma_0$.

\end{theorem}

\begin{proof}


We construct a sequence of approximate solutions as follows: let
\begin{equation*}
 u_\ell := \chi_{\ell}(r) e^{-i\omega t + i\lambda z - i\ell \phi} u_r(r) u_{\theta,\ell}(\theta)
\end{equation*}
Then recall that $u_\ell$ approximately solves the wave equation in the following sense:
\begin{equation*}
 || \Box_g u_\ell ||^2_{H^k} \lesssim (\chi_0)^2 e^{-C_k \ell \log \ell} 
\end{equation*}
Now, let $\tilde{u}_\ell$ be the solution to the wave equation with initial data on $\Sigma_0$ that agrees with the data for $u_\ell$.

We pick a sequence of open sets $U_\ell$ such that
\begin{equation}
 \begin{split}
  U_{\ell + 1} &\subset U_\ell \\
  \{ x \ \big| \ r(x) \leq (\mu_\ell)^{1+\epsilon} \} &\subset U_\ell \text{ for all sufficiently large } \ell
 \end{split}
\end{equation}
for some $\epsilon > 0$. In particular, if $U$ is any open set containing the submanifold $r = 0$ then $U_\ell \subset U$ for sufficiently large $\ell$. Indeed, if we wish then we can pick all the sets $U_\ell$ to be the same open set containing $r = 0$, independent of $\ell$. From Duhamel's principle (proposition \ref{proposition Duhamel}), taking norms in the region $U_{\ell, t} := U_\ell \cap \Sigma_t$ we find
\begin{equation}
 \begin{split}
  &||u_\ell - \tilde{u}_\ell ||_{H^1(U_{\ell,t})} + ||\partial_t u_\ell - \partial_t \tilde{u}_\ell||_{L^2(U_{\ell,t})} \\
  &\leq t \sup_{0 \leq s \leq t} \left( || P(t, x; s)\left(0, (g^{00})^{-1} \Box_g u_\ell \right) ||_{H^1(U_{\ell,t})} + ||\partial_t P(t, x; s)\left(0, (g^{00})^{-1} \Box_g u_\ell \right) ||_{L^2(U_{\ell,t})} \right)
 \end{split}
\end{equation}
Using the boundedness statement of theorem \ref{theorem boundedness two charge}, together with the Hardy inequality of lemma \ref{lemma Hardy}, the quantitative error estimate from the end of the previous section (equation \eqref{equation error}) and the fact that the norms of the quasimodes are independent of time, we find
\begin{equation}
\label{equation bound uniform decay v1}
 \begin{split}
  ||u_\ell - \tilde{u}_\ell ||_{H^1(U_{\ell,t})} + ||T (u_\ell - \tilde{u}_\ell) ||_{L^2(U_{\ell,t})}
  &\lesssim t \left( ||\Box_g u_\ell ||_{L^2(\Sigma_0)} \right) \\
  &\lesssim t (1 + (R \ell)^2) e^{-C_2 \ell \log \ell} ||u_\ell ||_{L^2(\Sigma_0)}
 \end{split}
\end{equation}
Where the factor of $(R \ell)^2$ comes from estimating the lower order term using the Hardy inequality. Note, however, that by taking $C_2$ to be slightly smaller we can absorb this contribution in the exponential factor. Now, we have
\begin{equation}
\label{equation bound uell in U}
 \begin{split}
  ||u_\ell ||_{L^2(\Sigma_0)} &= ||u_\ell ||_{L^2(U_{\ell, 0})} + ||u_\ell ||_{L^2(\Sigma_0 \setminus U_{\ell,0})} \\
  &\lesssim ||u_\ell ||_{L^2(U_{\ell, 0})} + e^{-\delta |\tilde{\omega}| \mu_\ell ^{(\frac{1}{2} + \epsilon)}} ||u_\ell ||_{L^2(\Sigma_0)}
 \end{split}
\end{equation}
Recall that we have a lower bound on $\tilde{\omega} \geq \tilde{\omega}_{\text{min}}$, independent of $\ell$. Hence, by taking $\ell$ sufficiently large we can make the coefficient of the second term on the right hand side of \eqref{equation bound uell in U} arbitrarily small, and so it can be absorbed by the left hand side. Thus, for sufficiently large $\ell$, we find
\begin{equation}
 ||u_\ell ||_{L^2(\Sigma_0)} \lesssim ||u_\ell ||_{L^2(U_{\ell, 0})}
\end{equation}
so that the norm of $u_\ell$ on the whole of $\Sigma_0$ can be bounded by its norm in the smaller region $U_\ell$.

Returning to equation \eqref{equation bound uniform decay v1}, and using the trivial bound
\begin{equation*}
  ||u_\ell ||_{L^2(U_{\ell, 0})} \leq  ||u_\ell ||_{H^1(U_{\ell, 0})}
\end{equation*}
along with the fact that norms of the quasimodes are time-independent, we find
\begin{equation}
 ||u_\ell - \tilde{u}_\ell ||_{H^1(U_{\ell,t})} + ||T (u_\ell - \tilde{u}_\ell) ||_{L^2(U_{\ell,t})}
 \lesssim t e^{-C_2 \ell \log \ell} \left( ||u_\ell ||_{H^1(U_{\ell,t})} + ||T u_\ell||_{L^2(U_{\ell,t})} \right)
\end{equation}
Hence, for all sufficiently large constants $C$, if we take $t$ to satisfy
\begin{equation*}
 t \leq C e^{C_2 \ell \log \ell}
\end{equation*}
then we can apply the reverse triangle inequality, and conclude that
\begin{equation}
\label{equation bound uniform decay v2}
  ||\tilde{u}_\ell ||_{H^1(U_{\ell,t})} + ||T\tilde{u}_\ell ||_{L^2(U_{\ell,t})} \gtrsim ||u_\ell ||_{H^1(U_{\ell,t})} + ||T u_\ell ||_{L^2(U_{\ell,t})} 
\end{equation}
Now, we have
\begin{equation*}
 \begin{split}
  ||u_\ell ||_{H^1(U_{\ell,t})} + ||T u_\ell ||_{L^2(U_{\ell,t})} 
  &= ||u_\ell ||_{H^1(U_{\ell,t})} + \tilde{\omega} || u_\ell ||_{L^2(U_{\ell,t})}
 \end{split}
\end{equation*}
while we also have
\begin{equation*}
 \begin{split}
  ||u_\ell ||_{H^1(\Sigma_t)} &= ||u_\ell ||_{H^1(U_{\ell,t})} + ||u_\ell ||_{H^1(\Sigma_t \setminus U_{\ell,t})} \\
  &\lesssim ||u_\ell ||_{H^1(U_{\ell,t})} + e^{-\delta |\tilde{\omega}| \mu_\ell ^{(\frac{1}{2} + \epsilon)}} ||u_\ell ||_{L^2(\Sigma_0)}
 \end{split}
\end{equation*}
so, for all sufficiently large $\ell$, for fixed $\tilde{\omega}_{\text{max}}$ we have
\begin{equation}
 ||u_\ell ||_{H^1(\Sigma_t)} + ||T u_\ell ||_{L^2(\Sigma_t)} \lesssim ||u_\ell ||_{H^1(U_{\ell,t})} + ||T u_\ell ||_{L^2(U_{\ell,t})}
\end{equation}
Substituting into equation \eqref{equation bound uniform decay v2} and using the fact that the norms of the quasimodes are independent of time, we find that, for all times $t \leq C e^{\ell \log \ell}$ we have the bound
\begin{equation}
 \begin{split}
  ||\tilde{u}_\ell ||_{H^1(U_{\ell,t})} + ||T \tilde{u}_\ell ||_{L^2(U_{\ell,t})} &\gtrsim ||u_\ell ||_{H^1(\Sigma_0)} + ||T u_\ell ||_{L^2(\Sigma_0)} \\
  &\gtrsim ||\tilde{u}_\ell ||_{H^1(\Sigma_0)} + ||T \tilde{u}_\ell ||_{L^2(\Sigma_0)}
 \end{split}
\end{equation}
where in the last line we have used the fact that the solutions to the wave equation $\tilde{u}_\ell$ and the quasimodes $u_\ell$ have identical initial data on $\Sigma_0$.

We now choose a sequence of times
\begin{equation}
 t_\ell := C e^{\ell \log \ell}
\end{equation}
for some sufficiently small constant $C$. Then we find that we have constructed a sequence of times $t_\ell \rightarrow \infty$, and a sequence solutions to the wave equation $\tilde{u}_\ell$, such that for some constant $c_1 > 0$ \emph{independent of} $\ell$ we have
\begin{equation}
\label{equation bound uniform decay v3}
\frac{ ||\tilde{u}_\ell ||_{H^1(U_{\ell,t})} + ||T \tilde{u}_\ell ||_{L^2(U_{\ell,t})} }{ ||\tilde{u}_\ell ||_{H^1(\Sigma_0)} + ||T \tilde{u}_\ell ||_{L^2(\Sigma_0)} } \geq c_1
\end{equation}

We may also consider replacing the denominator in equation \eqref{equation bound uniform decay v3} with higher order norms of the initial data. We claim that the following bound holds for the initial data:
\begin{equation}
\label{equation higher norms}
 \sqrt{E_0^{(1)}(u_\ell)} \lesssim || u_\ell ||_{H^2(\Sigma_0)} + || T u_\ell ||_{H^1(\Sigma_0)} \lesssim \ell \left( || u_\ell ||_{H^1(\Sigma_0)} + || T u_\ell ||_{L^2(\Sigma_0)} \right)
\end{equation}
The first inequality is clear from the definition of the higher order energies, while the second can be established at follows. Since $u_\ell$ is the initial data for a quasimode, we can exchange $z$ derivatives for factors of $\lambda$, which has been treated as some fixed constant in all of the above calculations. Additionally, we can exchange the $T$ derivative for a factor of $\tilde{\omega}$, which is actually bounded independently of $\ell$ by $\tilde{\omega}_{\text{max}}$. 

It is slightly more difficult to bound the higher order angular derivatives, and we need to use elliptic estimates. For this, we first note that the function
\begin{equation}
 \hat{u} := e^{-i\ell \phi} u_{\theta, \ell}(\theta)
\end{equation}
satisfies the eigenvalue equation
\begin{equation}
\label{equation eigenvalue on three sphere}
 -\slashed{\Delta}\hat{u} + \left( b^2(\tilde{\lambda}^2 - \tilde{\omega}^2)\frac{a^2}{(R_z)^2} \cos^2 \theta \right)\hat{u} = (\mu_\ell)^2 \hat{u}
\end{equation}
where $\slashed{\Delta}$ is the Laplacian on the 3-sphere, and $\hat{u}$ is interpreted as a function on the $3$-sphere by identifying the coordinates $(\theta, \phi, \psi)$ with the standard Hopf coordinates on the 3-sphere.

From equation \eqref{equation eigenvalue on three sphere} we find that we can estimate the first angular derivatives by the following: for sufficiently large $\mu_\ell$,
\begin{equation}
 \begin{split}
  \int_{\mathbb{S}^3} |\slashed{\nabla} \hat{u}|^2 \dVol_{\mathbb{S}^3} &\gtrsim (\mu_\ell)^2 \int_{\mathbb{S}^3} |\hat{u}|^2 \dVol_{\mathbb{S}^3} \\
  \Rightarrow ||\hat{u}||_{H^1(\mathbb{S}^3)} &\gtrsim |\mu_\ell| \, ||\hat{u}||_{L^2(\mathbb{S}^3)}
 \end{split}
\end{equation}
where the first line is obtained by multiplying \eqref{equation eigenvalue on three sphere} by $\hat{u}$ and integrating by parts. Similarly, from \eqref{equation eigenvalue on three sphere} we can estimate second angular derivatives by
\begin{equation}
 \begin{split}
  \int_{\mathbb{S}^3} \left| \slashed{\Delta}\hat{u} \right|^2 \dVol_{\mathbb{S}^3} &= \left( (\mu_\ell)^2 - b^2(\tilde{\lambda}^2 - \tilde{\omega}^2)\frac{a^2}{(R_z)^2} \cos^2 \theta \right)^2 |\hat{u}|^2 \\
  &\lesssim (\mu_\ell)^4 ||\hat{u}||^2_{L^2(\mathbb{S}^3)}
 \end{split}
\end{equation}
however, integrating by parts, the left hand side of this equation can be written as
\begin{equation}
 \int_{\mathbb{S}^3} \left| \slashed{\Delta}\hat{u} \right|^2 \dVol_{\mathbb{S}^3} = \int_{\mathbb{S}^3} \left( |\slashed{\nabla}^2 \hat{u}|^2 + \slashed{R}^{\mu\nu} (\slashed{\nabla}_\mu \hat{u}) (\slashed{\nabla}_\nu \hat{u}^*) \right) \dVol_{\mathbb{S}^3}
\end{equation}
where $\hat{u}^*$ denotes the complex conjugate of $\hat{u}$, and $\slashed{R}^{\mu\nu}$ is the Ricci curvature tensor of $\mathbb{S}^3$. For the 3-sphere, we have
\begin{equation*}
 \slashed{R}^{\mu\nu} = 2\slashed{g}^{\mu\nu}
\end{equation*}
where $\slashed{g}$ is the metric on the unit 3-sphere, and so
\begin{equation}
 \int_{\mathbb{S}^3} \left( |\slashed{\nabla}^2 \hat{u}|^2 + |\slashed{\nabla} \hat{u}|^2 \right) \lesssim (\mu_\ell)^4 ||\hat{u}||^2_{L^2(\mathbb{S}^3)}
\end{equation}
which, combined with the previous calculations, leads to the desired conclusion
\begin{equation}
 || \hat{u} ||_{H^2(\mathbb{S}^3)} \lesssim (\mu_\ell) ||\hat{u}||_{H^1(\mathbb{S}^3)}
\end{equation}
In other words, additional angular derivatives may be exchanged for factors of $\mu_\ell$.

Finally, we need to bound second derivatives with respect to $r$. For this, we use the fact that the quasimode $u_\ell$ gives an approximate solution to the wave equation:
\begin{equation*}
 ||\Box_g u_\ell ||_{L^2(\Sigma_0)} \lesssim e^{-C_2 \ell \log \ell} ||u_\ell ||_{L^2(\Sigma_0)}
\end{equation*}
In the inverse metric, $g^{rr}$ is bounded above and below, so we can exchange second derivatives with respect to $r$ for other second derivatives and lower order terms. We have already explained how to estimate these, with the exception of second time derivatives. However, for the quasimode $u_\ell$ (but \emph{not} for the actual solution to the wave equation $\tilde{u}_\ell$) we have
\begin{equation*}
 T^2 u_\ell = -(\tilde{\omega})^2 u_\ell
\end{equation*}
so these second times derivatives are easy to bound.

Higher derivatives can then be bounded by using similar elliptic estimates for the angular derivatives, and higher $r$ derivatives can be bounded in terms of other derivatives by differentiating the (approximate) wave equation satisfied by $u_\ell$, leading to the following estimate: for all $k \geq 1$, for all sufficiently small $C_k > 0$, at all times $t \leq C_k e^{C_k\ell \log \ell}$ there exists a constant $c_k > 0$ such that
\begin{equation}
 \frac{ ||\tilde{u}_\ell ||_{H^1(U_{\ell,t})} + ||T \tilde{u}_\ell ||_{L^2(U_{\ell,t})} }{ ||\tilde{u}_\ell ||_{H^k(\Sigma_0)} + ||T \tilde{u}_\ell ||_{H^{k-1}(\Sigma_0)} } \geq c_k (\ell)^{k-1}
\end{equation}
We now pick the times $t_\ell$ to be
\begin{equation}
 t_\ell := C_k e^{C_k \ell \log \ell}
\end{equation}
in which case $\ell$ can be written in terms of $t_\ell$ as
\begin{equation}
 \ell = \exp\left( W_0 \left( \frac{1}{C_k} \log \left( \frac{t_\ell}{C_k} \right) \right) \right)
\end{equation}

Finally, we note that the initial data we have constructed are is Schwartz. Indeed, the above calculations give bounds on the derivatives of $u_\ell$ and $T u_\ell$. We can also show that these sequences of functions lie in all polynomially weighted Sobolev spaces; indeed, we have, 
\begin{equation}
 ||r^\alpha u_\ell ||_{H^k(\Sigma_0)} \leq C_\alpha ||u_\ell||_{H^k(\Sigma_0)}
\end{equation}
for some constants $C_\alpha$ which, importantly, are \emph{independent} of $\ell$. Note that, although the initial data induced by each function $u_\ell$ is compactly supported, the region of support depends on $\ell$, indeed, $u_\ell$ is supported in the region $r \leq R\mu_\ell$. However, recall that the norm $u_\ell$ is \emph{exponentially} suppressed in the parameter $\ell$ outside of any region of fixed size, which easily allows us to prove the inequality above and hence to prove the theorem. Note, however, that we cannot use our sequence of functions if we wish to take the supremum over initial data which is compactly supported in some given region, since the sequence of initial data that we construct will eventually be supported outside this region.

\end{proof}

Note that we can repeat the above calculation in the three charge case, and in fact, it is slightly easier since the minimum of the effective potential is then at some nonzero value of $r$. This might seem redundant in the light of the results of section \ref{section non decay}, however, note that the initial data constructed in that section \emph{must} have nontrivial dependence on the $z$ direction. We can see this from the fact that the energy associated with the null vector field $V = T + Z$ is non-negative, but if $u$ has trivial dependence on the $z$ coordinate then $(Zu) = 0$, so the energy associated with $T$ is non-negative as well. Hence, if we wish to restrict to waves with trivial dependence on the $z$ direction, then we can obtain the same results in both the two and three charge microstate geometries.

\section*{Acknowledgements}
The author is very grateful to Harvey Reall, Felicity Eperon and Jorge Santos for numerous helpful discussions, and also to Harvey Reall and Mihalis Dafermos for useful comments on the manuscript. This work was partially supported by the European Research Council Grant No. ERC-2011-StG 279363-HiDGR.

\sloppy
\printbibliography

\end{document}